\documentclass{article}
\usepackage[utf8]{inputenc}
\usepackage{authblk}
\usepackage{setspace}
\usepackage[margin=1.25in]{geometry}
\usepackage{graphicx}
\graphicspath{ {./figures/} }
\usepackage{subcaption}
\usepackage{amsmath}
\usepackage{hyperref}
\usepackage{amsmath,amssymb}
\usepackage{algorithm}
\usepackage[noend]{algpseudocode}
\usepackage{float}  % 引入float包

\usepackage{enumitem}
\usepackage{graphicx}
\usepackage{amsmath, amssymb, amsthm}
\usepackage{amsthm}
\usepackage{physics}
\usepackage{mathtools}
\usepackage{bm}
\usepackage{geometry}
\usepackage{booktabs}
\usepackage{microtype}

\newtheorem{theorem}{Theorem}[section]
\newtheorem{definition}{Definition}[section]
\newtheorem{assumption}{Assumption}[section]

\newtheorem{remark}{Remark}[section]

\usepackage[style=nejm, 
citestyle=numeric-comp,
sorting=none]{biblatex}
\title{Early warning prediction: Onsager-Machlup vs Schrödinger }

\author[1,6,7]{Xiaoai Xu}
\author[2,3]{Yixuan Zhou}
\author[5]{Xiang Zhou}
\author[6,7]{Jingqiao Duan}
\author[2,3,4*]{Ting Gao}

\affil[1]{School of Mathematics and Information Science, Guangzhou University, Guangzhou 510000, China.}
\affil[2]{School of Mathematics and Statistics, Huazhong University of Science and Technology, Wuhan 430074, China.}
\affil[3]{Center for Mathematical Science, Huazhong University of Science and Technology, Wuhan 430074, China}
\affil[4]{Steklov-Wuhan Institute for Mathematical Exploration, Huazhong University of Science and Technology, Wuhan 430074, China}
\affil[5]{Department of Mathematics, City University of Hong Kong, Kowloon, Hong Kong SAR.}
\affil[6]{School of Sciences, Great Bay University, Dongguan 523000, China.}
\affil[7]{Guangdong Provincial Key Laboratory of Mathematical and Neural Dynamical Systems, Great Bay University, Dongguan 523000,  China.}
\affil[*]{Address correspondence to: tgao0716@hust.edu.cn}

\date{}

\begin{document}

\maketitle

%%%%%% Abstract %%%%%%
\begin{abstract}

Predicting critical transitions in complex systems, such as epileptic seizures in the brain, represents a major challenge in scientific research. The high-dimensional characteristics and hidden critical signals further complicate early-warning tasks. This study proposes a novel early-warning framework that integrates manifold learning with stochastic dynamical system modeling. Through systematic comparison, six methods including diffusion maps (DM) are selected to construct low-dimensional representations. Based on these, a data-driven stochastic differential equation model is established to robustly estimate the probability evolution scoring function of the system. Building on this, a new Score Function (SF) indicator is defined by incorporating Schrödinger bridge theory to quantify the likelihood of significant state transitions in the system. Experiments demonstrate that this indicator exhibits higher sensitivity and robustness in epilepsy prediction, enables earlier identification of critical points, and clearly captures dynamic features across various stages before and after seizure onset. This work provides a systematic theoretical framework and practical methodology for extracting early-warning signals from high-dimensional data.
\end{abstract}

%%%%%% Main Text %%%%%%
%A few notes on the main text:  
%Three heading levels are permitted. Only the headings listed are permitted as Level 1 headings. Authors are encouraged to indicate the level of each heading using a unique format. For example, 
%\textbf{LEVEL 1 IN BOLD CAPS}
%\textbf{Level 2 in bold}
%\textit{Level 3 in italics}

\section{INTRODUCTION}
\hspace*{1em}Complex systems are prevalent in various fields such as ecology, climate science, finance, and neuroscience. These systems often exhibit critical transition phenomena, where a slight disturbance can lead to sudden and significant changes in the system's state \cite{scheffer2001catastrophic}. In ecological systems, lakes may rapidly change from a clear state with rich biodiversity to a turbid state with algal blooms due to factors such as eutrophication, resulting in a large number of fish deaths and the disruption of ecological balance \cite{scheffer2001catastrophic}. In medicine, the occurrence of a disease is usually not a smooth process, but rather a sudden transition that occurs when a critical point is reached \cite{liu2011early,chen2012detecting}. These critical points represent the key states of disease transition, such as the transition from health to disease \cite{deb2022identifying}. In the social-economic field, financial markets may collapse suddenly, bringing huge impacts to society \cite{fabiani2024task}.

A critical point is the point at which the system behavior undergoes a qualitative change when the system parameters reach a certain threshold, at which point the system becomes highly sensitive to perturbations. Critical points can be divided into three categories: bifurcation-related critical points (B-tipping), noise-induced critical points (N-tipping), and rate-dependent critical points (R-tipping) \cite{ashwin2012tipping}. N-tipping refer to the phenomenon wherein a system departs from the neighborhood of a quasi-static attractor driven by stochastic fluctuations noise. This process is theoretically independent of parameter changes crossing bifurcation points; noise alone can trigger an abrupt state transition in the system, which is our focus for this paper. Traditional analytical frameworks predominantly concentrate on the examination of specific statistical features, such as variance amplification and changes in autocorrelation, among others \cite{scheffer2009early}. Nevertheless, the analysis of complex systems characterized by high-dimensional data and intricate dynamic properties presents substantial challenges to these conventional methods.

Recognizing the challenges posed by high-dimensional data, researchers have dedicated extensive research efforts to developing robust analytical frameworks and methodological innovations. Zheng et al. formalize global resilience in tropical forests, within a dynamical systems framework, as the mean exit time. Researchers' analysis with stochastic differential equations demystified the transition behavior of unstable critical points in multi-stable systems, offering a formal dynamical method for tipping point analysis \cite{zheng2023mean}. The study of critical points in complex networks necessitates the consideration of network topology and other characteristics, which differ from the analysis of dynamical systems in low dimension. The Early Warning Signal Network (EWSNet) framework provides a systematic methodology for detecting precursory signals prior to critical transitions or systemic collapse through comprehensive analysis of network topological configurations and temporal dynamics. Complex networks are applied to analyze gene regulatory networks, protein-protein interaction networks, and other biological systems. EWSNet can help identify key nodes and early warning signals in the occurrence and development of diseases \cite{olgun2024novel,zhong2025identification}. The Dynamic Network Biomarker (DNB)  method is a commonly used approach based on bifurcation theory to detect critical points of phase transitions in complex systems from data \cite{tao2023data}. The Master Stability Function (MSF) method can be utilized to analyze the loss of synchronization stability in complex networks, thereby predicting abrupt transitions in network dynamical behavior. There are also related studies on constructing predictive models using advanced technologies such as deep learning. For example, the EWSNet detects early warning signals of critical transitions by training a deep neural network using a simple theoretical model, and this network can distinguish the types of critical transitions \cite{deb2022machine}. Feng et al. utilize latent stochastic dynamical systems to extract early warning signals from high-dimensional observational data \cite{feng2024early}. In addition, there are deep learning algorithms based on two-dimensional convolutional neural networks (CNN) and long short-term memory networks (LSTM) for predicting the emergence of critical points in dynamic systems \cite{zhuge2024deep}. Zhang et al. propose a two-level optimization framework that combines, utilizing the predicted probability by the meta-network as an early warning indicator \cite{zhang2025early}. This approach captures critical signals before seizure onset than variance indicators, providing an effective solution for early warning of epilepsy patients \cite{zhang2025early}.

The analysis of rare events, characterized by an exceedingly low probability of occurrence, presents a significant computational challenge. Conventional Monte Carlo simulations, which often require sampling on the order of thousands of trials to capture such events, can become computationally prohibitive. The theory of large deviations provides a rigorous framework to describe their exponentially small likelihood. Within it, the Freidlin-Wentzell theory characterizes the pathwise probabilities of stochastic systems under weak noise \cite{2008The}. The Onsager-Machlup action provides a powerful tool for quantifying the probability of any path, effectively measuring its deviation from the most probable transition pathway \cite{feng2024early}. In this context, the Schrödinger bridge problem reformulates the search for a connecting stochastic process as the minimization of the OM functional. Modern numerical solutions employ diffusion models and score matching. Building on this, we propose a new ``score-function indicator" for early warning. By learning the dynamical score function, it detects early signals of a system approaching a critical, rare-event state.

The contribution of this paper is in the following aspects:

\begin{itemize}
    \item We construct a theoretical framework that integrates manifold learning with stochastic dynamics to examine the impact of dimensionality reduction.
    \item We propose a novel early-warning indicator via the score function.
    \item We compare and analyze the OM-functional indicator and the score-function indicator using real-world data.
\end{itemize}

The remaining structure of this paper is as follows: In Section 2, we first introduce the application background of manifold learning in high-dimensional data dimensionality reduction, laying the foundation for subsequent analysis. Subsequently, we present the relevant theories of the Onsager-Machlup action and the Schrödinger bridge problem. Based on these theories, we construct two warning indicators in Section 3—the Onsager-Machlup indicator and the Score Function indicator—and conduct a systematic comparative analysis of the two. To verify the effectiveness of the method, experimental validation is conducted in Section 4 using real electroencephalogram data. Furthermore, we perform a formal error analysis of the warning potential in Section 5. Finally, Section 6 summarizes the entire work and draws conclusions.

%%%%% 第二部分 %%%%%%
\section{Background Knowledge}

\subsection{Dimensionality Reduction via Manifold Learning}

\hspace*{1em} Manifold learning is a machine learning method used to discover low-dimensional structures within high-dimensional data. These methods assume that the data actually lies on or near a low-dimensional manifold embedded in a high-dimensional space. Through nonlinear transformation, manifold learning algorithms can find a low-dimensional representation (embedding) while retaining the most important structure of the original data \cite{Zheng2009}\cite{meilă2023manifold}.

Manifold learning can be viewed as an attempt to generalize linear frameworks such as Principal Component Analysis (PCA), enabling them to capture nonlinear structures in data. PCA is a linear dimensionality reduction technique that transforms the original data into a set of linearly uncorrelated variables through orthogonal transformation. Kernel Principal Component Analysis (Kernel PCA) utilizes kernel methods to non-linearly map the data into a high-dimensional space and performs PCA in this space, aiming to find a low-dimensional embedding suitable for linear classification of the originally non-linearly separable data. Isometric Mapping (Isomap) is one of the earliest manifold learning methods. This algorithm can be regarded as an extension of Multidimensional Scaling (MDS) or Kernel PCA. Isomap seeks to find a low-dimensional embedding that preserves the geodesic distances between all points \cite{tenenbaum2000global}. Locally Linear Embedding (LLE) aims to preserve the local neighborhood distance relationships of the original data by seeking a projection of the data in a low-dimensional space \cite{roweis2000nonlinear}. Neighborhood Preserving Embedding (NPE) is a linearly approximated manifold learning method that attempts to preserve the local neighborhood structure of the data after dimensionality reduction. NPE can be considered the linear version of the manifold learning method LLE. Spectral Embedding is a nonlinear dimensionality reduction method that preserves the local structure of data through the spectral decomposition of the graph Laplacian matrix \cite{belkin2003laplacian}. Orthogonal Locality Preserving Projection (OLPP) is an improved version of Locality Preserving Projection (LPP), introducing orthogonal constraints to achieve better performance. LPP is a linear dimensionality reduction method designed to preserve the local neighborhood relationships between data points, while OLPP orthogonalizes the projection directions, making the eigenvectors orthogonal to avoid redundant information. MDS is a technique that reduces dimensionality by preserving the distances or the order of distances between original data points in a low-dimensional space  \cite{mair2022more}. T-Distributed Stochastic Neighbor Embedding (T-SNE) is a nonlinear dimensionality reduction technique that preserves the local structure of data by matching probability distributions in high-dimensional and low-dimensional spaces, making it suitable for visualizing and clustering high-dimensional data \cite{belkina2019automated}. Laplacian Eigenmaps is a graph-based dimensionality reduction method aimed at preserving the local neighborhood structure of the data. Diffusion Map (DM) is a nonlinear dimensionality reduction method that utilizes the diffusion process of random walks to convert spatial distances into transition probabilities. In manifold learning, OLPP and NPE incorporate orthogonal constraints and linearized improvements, respectively, balancing discriminability with computational efficiency while preserving local structures. Kernel PCA elegantly captures global nonlinear features of data through the kernel trick. In contrast, Isomap and DM are more adept at revealing complex global geometries and multi-scale nonlinear structures. These methods perform particularly well when handling complex manifold shapes.

Although a wide variety of manifold learning methods exist, not all are equally suitable for the specific task of early warning. The core of early warning lies in capturing subtle and leading signals of critical transitions in a system from high-dimensional, nonlinear, and often noisy data. Based on this, the selection criteria in this study primarily consider the robustness of the methods, their ability to capture the intrinsic geometric structure of the data, and their sensitivity to temporal dynamics or local structures. Therefore, this study adopts six representative manifold learning methods, namely OLPP, DM, Isomap, PCA, Kernel PCA, and NPE, to conduct the research.

\subsection{Onsager-Machlup Action Functional and Schrödinger Bridge Problem}

\subsubsection{Onsager-Machlup Action Functional}

\hspace*{1em}In the  stochastic differential equation(SDE)

\begin{equation}
    dX(t) = b(X(t))dt + \sigma (X(t))dW(t),0 \leq t \leq T,
\end{equation}

\noindent what we are interested in is studying the transition phenomenon between its two metastable states. Specifically, $b(X(t)): \mathbb{R}^d \to \mathbb{R}^{d}$ represents a drift, $\delta: \mathbb{R}^d \to \mathbb{R}^{d \times u}$denotes a d×u matrix-valued function, and W denotes a Brownian motion on $R^u$.

The core idea of the Onsager-Machlup theory is to assign a probability density to these continuous trajectories (rather than just a certain end state). More precisely, it estimates the probability of the true solution $X(t)$ occurring within a minimal `tube' $\sigma$ around a given smooth reference trajectory $z(t)$ \cite{guo2024deep}

\begin{equation}
    \mathrm{P}\left\{\| X - z \|_T \leq \delta\right\} \propto \mathrm{C}(\delta, \mathrm{T})\exp\left\{-\mathrm{S}_T^{\mathrm{OM}}(z, \dot{z})\right\},
\end{equation}

\noindent where $\delta$ is positive and sufficiently small, the norm of $\|\cdot\|_T$ is defined as $\|a(t)\|_T := \sup_{0 \leq t \leq T} |a(t)|$ , and $C(\delta, T)$ is a normalization constant related to the tube radius $\delta$ and time T, but it is independent of the specific reference path z. The $exp\left\{-\mathrm{S}_T^{\mathrm{OM}}(z, \dot{z})\right\}$ is the core part of probability. It explicitly relies on the path z we choose. It indicates that the probability of the path decays exponentially as $\mathrm{S}_T^{\mathrm{OM}}$ increases. Therefore, $\mathrm{S}_T^{\mathrm{OM}}$ is called the action functional.It is akin to a `cost function' or `energy function' defined for paths, where paths with higher action functional values have lower probabilities of occurrence, and paths with lower action functional values have higher probabilities of occurrence. Its functional role is \cite{feng2024early}

\begin{equation}
    S_T^{\mathrm{OM}}(z, \dot{z}) = \frac{1}{2} \int_0^T \left[ \frac{|\dot{z} - b(z)|^2}{\sigma^2} + \nabla \cdot b(z) \right] \mathrm{d}t.
\end{equation}

\subsubsection{Schrödinger Bridge Problem}

\hspace*{1em}The Schrödinger Bridge is a mathematical framework used to solve the dynamical optimal transport problem between two probability distributions. Specifically, given an initial probability measure and a final probability measure, the Schrödinger bridge problem aims to find the time evolutionary probability density connecting given initial and final probability measures that minimizes the relative entropy.

\begin{definition}
Definition 1(Relative Entropy) \cite{leonard2013survey}: Let r be a certain  $\sigma$-finite positive measure on space y. The relative entropy of probability measure p with respect to r is loosely defined as
\begin{equation}
    H(P \mid R) := \int_Y \log(dP/dR) \, dP \in (-\infty, \infty], \quad P \in \mathrm{P}(Y).
\end{equation}
If P is much smaller than R, the value is infinite. The smaller $H(P\mid R)$ is, the more similar the new process P is to the original Brownian motion R, and thus the higher its probability of occurrence (in a large number of independent particle systems, its realization frequency is the highest).
Therefore, finding the most likely process is transformed into a variational problem of minimizing the relative entropy $H(P\mid R)$ under the marginal constraints $P_0=\mu_0$ and $P_1=\mu_1$.
\end{definition}

\begin{definition}
    Definition 2 (Dynamic Schrödinger's Problem) \cite{leonard2013survey}: The dynamic Schrödinger problem associated with a reference path measure $R \in \mathcal{M}_+(\Omega)$ is the following entropy minimization problem 

\begin{equation}
      H(P \mid R) \to \min; \quad P \in \mathrm{P}(\Omega): P_0 = \mu_0, \, P_1 = \mu_1 
\end{equation}
where $\mu_0, \mu_1 \in \mathrm{P}(\mathcal{X})$ are prescribed initial and final marginals.
\end{definition}

Its general form is as follows \cite{de2021diffusion}

\begin{equation}
    \Pi^* = \arg\min \left\{ KL(\Pi \mid P) : \Pi \in P_{N+1}, \Pi_0 \in P_{\text{data}}, \Pi_1 \in P_{\text{prior}} \right\},
\end{equation}

\noindent where \( KL(\Pi \| P) \) represents the Kullback-Leibler divergence of the path distribution \( \pi \) relative to the reference distribution \( P \), measuring the difference between the two. \( \mathcal{P}_{N+1} \) is the set of path probability distributions defined over \( N+1 \) time steps. The distributions \( \pi_0 \) and \( \pi_1 \) are the marginal distributions of the path at the initial time \( t=0 \) and the terminal time \( t=1 \), respectively, following the data distribution \( P_{\text{data}} \) and the prior distribution \( P_{\text{prior}} \). By applying an iterative optimization method with relaxed constraints to the above equation, the Iterative Proportional Fitting (IPF) is obtained. Furthermore, the original optimization of the joint probability distribution in each round is decomposed into optimizing a series of conditional distributions between adjacent time steps, thereby decomposing into a series of forward and backward conditional distributions. During the training process, the forward and backward trajectories are optimized alternately \cite{de2021diffusion}

\begin{equation}
\begin{aligned}
\Pi^{2n+1} &= \arg\min \left\{ KL\left( \Pi_{t|t+1} \mid \Pi_{t|t+1}^{2n} \right) : \Pi \in P_{N+1}, \, \Pi_N \in P_{\text{prior}} \right\} ,\\
\Pi^{2n+2} &= \arg\min \left\{ KL\left( \Pi_{t+1|t} \mid \Pi_{t+1|t}^{2n+1} \right) : \Pi \in P_{N+1}, \, \Pi_0 \in P_{\text{data}} \right\},
\end{aligned}
\end{equation}

The Schrödinger bridge provides a generative model for time series. In this framework, the evolution and generation of new time series samples can be described and simulated via a stochastic differential equation (SDE) \cite{hamdouche2023generative}. Specifically, Chen et al. demonstrate that the optimal solution to this problem takes the form of the SDE \cite{Chen2021LikelihoodTO}

\begin{equation}
\begin{aligned}
\mathrm{d}{\bf X}_{t}&=\left(f({\bf X}_{t},t)+g^{2}(t)\nabla\log\boldsymbol{\Psi}({\bf X}_{t},t)\right)\mathrm{d}t+g(t)\mathrm{d}\mathbf{W}_{t}, \, X_{0}\sim p_{\mathrm{data}},\\
\mathrm{d}{\bf X}_{t}&=\left(f({\bf X}_{t},t)-g^{2}(t)\nabla\log\boldsymbol{\hat{\Psi}}({\bf X}_{t},t)\right)\mathrm{d}t+g(t)\mathrm{d}\bar{\mathbf{W}}_{t}, \, X_{T}\sim p_{\mathrm{prior}},
\end{aligned}
\end{equation}

\noindent where $\mathbf{W}_{t}$ is Wiener process and $\bar{\mathbf{W}}_{t}$ is its time reversal. $\boldsymbol{\Psi},\boldsymbol{\hat{\Psi}}\in\mathcal{C}^{2,1}\left([0,T],\mathbb{R}^{d}\right)$ are time-varying energy potentials that constrained by the interconnected PDEs \cite{zhang2024action}

\begin{equation}
\begin{cases}
\frac{\partial\boldsymbol{\Psi}}{\partial t}=-\nabla_{x}\boldsymbol{\Psi}^{\top}f-\frac{1}{2}\mathrm{Tr}\left(g^{2}\nabla_{x}^{2}\boldsymbol{\Psi}\right), \\
\frac{\partial\boldsymbol{\hat{\Psi}}}{\partial t}=-\nabla_{x}\cdot\left(\boldsymbol{\hat{\Psi}}f\right)+\frac{1}{2}\mathrm{Tr}\left(g^{2}\nabla_{x}^{2}\boldsymbol{\hat{\Psi}}\right),
\end{cases}
\end{equation}

\noindent such that $\boldsymbol{\Psi}(x,0)\boldsymbol{\Psi}(x,0)=p_{\mathrm{data}},\boldsymbol{\Psi}(x,T)\boldsymbol{\Psi}(x,T)=p_{\mathrm{prior}}.$

\subsubsection{Score Matching Method}
\hspace*{1em}Applying the six selected dimensionality reduction methods to processed data can effectively predict critical points, significantly aiding our early warning efforts (see Experiment 4). We then learn from the reduced-dimensionality data, derive its stochastic differential equation through data-driven methods, and ultimately obtain the Schrödinger bridge equation by using fractional matching to derive the scoring function term S(t).

The methods for obtaining the scoring function S(t) include score matching (SM), sliced score matching (SSM), Physics-Informed Score Matching Networks (Score-PINN), denoising score matching, etc \cite{song2020sliced,hu2024score,olga2021denoising}. 

Score Matching (SM) is a key technique for estimating parameters in unnormalized statistical models. It circumvents the computationally challenging partition function (normalizing constant) required in maximum likelihood estimation by minimizing the discrepancy between the model’s score function and the true data score function, thereby enabling parameter inference without evaluating the intractable normalizing constant. However, Score Matching encounters significant computational bottlenecks when applied to high-dimensional data or complex models, as its objective depends on the second-order derivative (Hessian matrix) of the log-density, whose computation and storage become prohibitive in high-dimensional settings. Although subsequent approaches, such as Score-PINN, have partially alleviated these high-dimensional modeling difficulties, they still suffer from issues like high memory demand and suboptimal computational efficiency. In contrast, the SM framework itself offers a strong theoretical foundation. It directly optimizes the score-based objective, often delivering good computational efficiency and faster convergence when model structure and dimensionality are moderate, thereby maintaining a well-balanced trade-off between theoretical rigor and practical applicability.

The goal of Score Matching is to minimize the following loss function \cite{hu2024score}

\begin{equation}
    L_{\text{SM}}(\theta) = \mathbb{E}_{t \sim \text{Unif}[0,T]} \mathbb{E}_{\boldsymbol{x} \sim p_t(\boldsymbol{x})} \left[ \lambda(t) \left\| \boldsymbol{s}_t(\boldsymbol{x}; \theta) - \nabla_{\boldsymbol{x}} \log p_t \left( \boldsymbol{x} \right) \right\|^2 \right].
\end{equation}

The sample $\boldsymbol{x} \sim p_t(\boldsymbol{x})$ can be obtained via SDE discretization. Unif[0, T] denotes the uniform time sampling between 0 and T. Consequently, we can obtain the score function $ \boldsymbol{s}_t(\boldsymbol{x}; \theta) = \boldsymbol{s}_t(x)$ via optimizing $\theta$ in the loss function.

The specific method can be seen in the Figure \ref{fig: Overall Framework Flowchart} below.

\begin{figure}[htbp]  % 添加位置参数
  \centering
  \includegraphics[width=0.75\textwidth]{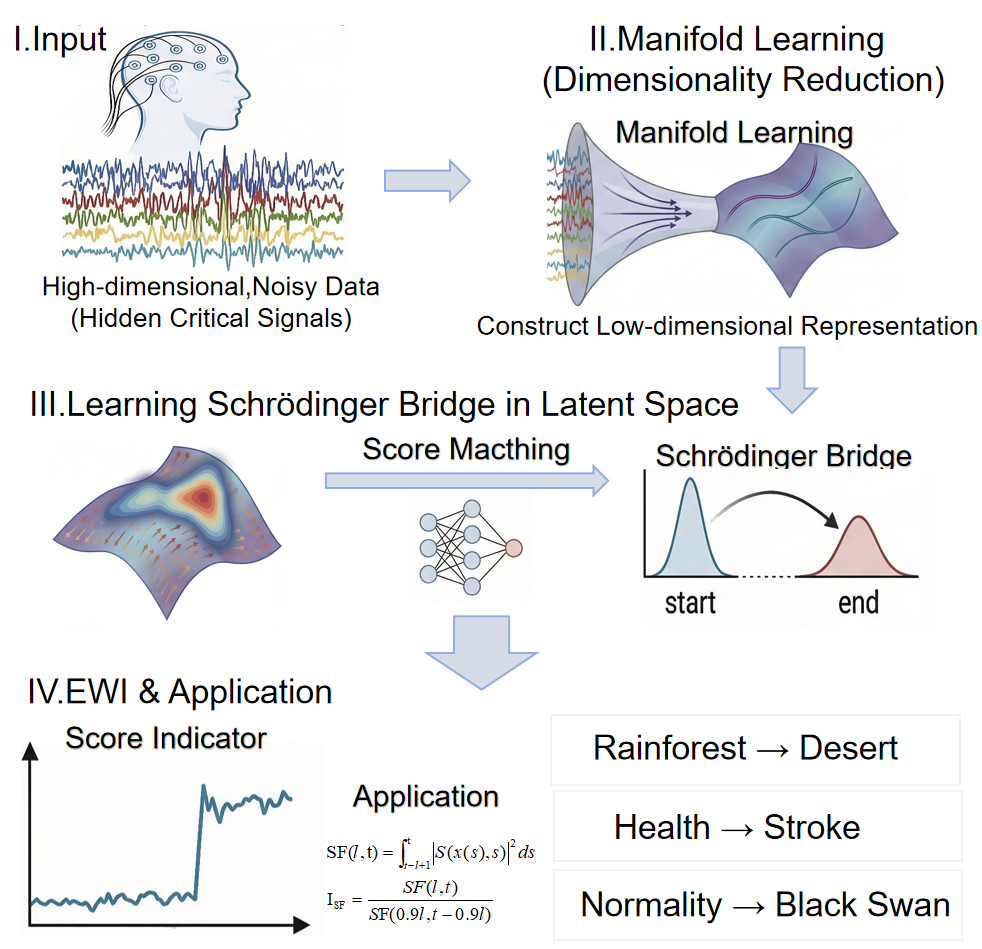}
  \caption{\small Overall Framework Flowchart}
  \label{fig: Overall Framework Flowchart}
\end{figure}

% \begin{algorithm}[H]
% \caption{Early Warning System via Manifold Learning and Score Matching}
% \label{alg:early_warning}

% \begin{algorithmic}
% \State \textbf{Input:} time series $X \in \mathbb{R}^{d \times n}$

% \State \textbf{Output:} Warning indicator, visualization results

% \State \textbf{Step 1: Dimensionality Reduction}
% \State $\phi, \Lambda \leftarrow \text{Manifold Learning}(X, \text{method}=\text{`DDM', `DM',` kernel PCA',`t-SNE', `MDS' and others})$

% \State \textbf{Step 2: SDE Learning}  
% \State $\theta_\mu, \theta_\sigma \leftarrow \text{Learn SDE}(\phi, \Delta t=1/16)$

% \State \text{Learned SDE: } 
% $$dX_t = \mu(X_t; \theta_\mu)dt + \sigma(X_t; \theta_\sigma)dW_t$$
% %$$dX_t = \mu(X_t; \theta_\mu)dt + \sigma(X_t; \theta_\sigma)dL_t$$

% \State \textbf{Step 3: Score Function Learning (Schrödinger Bridge)}

% \State $s_t(x; \theta) \leftarrow \text{Score Matching}(\phi, \theta_\mu, \theta_\sigma, T_{\text{max}}, n_{\text{steps}})$

% \State \text{Learn: } $s_t(x; \theta) \approx \nabla_x \log p_t(x)$
% \State \text{via: } $\theta = \arg \min_\theta \mathbb{E}[\|s_t(x; \theta) - \nabla_x \log p_t(x)\|^2]$

% \State \textbf{Step 4: Warning Indicator}
% \State $I_{\text{warning}}, T_{\text{axis}} \leftarrow \text{Compute Warning Indicator}(\phi, s_t(\cdot; \theta), \text{window\_size})$
% \State \text{Indicator: } $I(t) = \int_{\text{window}} |s(x, \tau)|^2  d\tau$

% % \State \textbf{Output}
% \State \textbf{return} $\{\phi, \Lambda, \theta_\mu, \theta_\sigma, s_t(\cdot; \theta), I_{\text{warning}}, T_{\text{axis}}\}$
% \end{algorithmic}
% \end{algorithm}

\section{Early Warning Indicators}

\subsection{Onsager-Machlup Indicator}
\hspace*{1em}Let the time series $x(t)$ be a random variable controlled by a stochastic differential equation in latent space. Construct two neural networks $\mu_\theta$ and $\eta_\theta$ to estimate $\mu$ and $\sigma$, with their inputs and weights represented by basis functions and coefficients, respectively. Here, we define the Onsager-Machlup (OM) function according to Eq(3) \cite{feng2024early}

\begin{equation}
OM(l, t) = \frac{1}{2} \int_{t - l + 1}^{t} \left[ \frac{|\dot{x}(s) - \mu_{\theta}(x(s))|^2}{\eta_{\theta}^2} + \nabla \cdot \mu_{\theta}(x(s)) \right] ds,
\end{equation}
where $\dot{z}$ is the derivative function written as $\dot{x}(s) \approx \frac{x(s + ds) - x(s - ds)}{2ds}$ with the time step $ds$.

The OM indicator is formulated as follows

\begin{equation}
    I_{OM} = \frac{OM(l,t)}{OM(0.9l,t-0.9l)}
\end{equation}
%到时候补充详细一点
\subsection{Score Function Indicator}
\hspace*{1em}The degree of imbalance can be measured by the generation of entropy. Entropy is defined as

\begin{equation}
    S = - \int P \ln P \, dx.
\end{equation}

For continuous dynamical systems, statistical behavior is governed by the probability density distribution \( P(x,t) \). The temporal evolution of the probability density is determined by the Fokker-Planck equation \cite{risken1989fokker}.

Inspired by entropy production and the OM functional, we propose an early-warning indicator for the Schrödinger bridge, which we call the Score Function (SF) Indicator

\begin{equation}
    \text{SF(l,t)} = E_x \int_{t-l+1}^t \|\nabla_{x(s)} \ln \varphi\|^2 ds = \int_{t-l+1}^t \|S(s)\|^2 ds. 
\end{equation}

The SF indicator is formulated as follows

\begin{equation}
    I_{SF} = \frac{SF(l,t)}{SF(0.9l,t-0.9l)}
\end{equation}

\subsection{Comparison of Two Indicators}

\hspace*{1em}The SF indicator is essentially identical to the OM functional indicator. Specifically, the SF indicator is derived by substituting the score function obtained through training into the Schrödinger bridge equation Eq(18), and then incorporating it into the first term—namely, the large deviation term—of the OM functional Eq(11), thereby resulting in a new indicator.

In the early warning analysis of critical transitions in complex systems, the OM functional indicator and the SF indicator offer distinct monitoring approaches from the perspectives of probability functionals and thermodynamics, respectively. The OM functional, based on the principle of least action, quantifies the likelihood of state transitions through path probability functionals. Its analysis relies on systems with well-defined stochastic differential equation models (such as chemical or neuronal dynamics with known drift and diffusion terms) and often exhibits steep extrema or bifurcation behavior in the functional before critical points, indicating the system’s tendency toward a new state. In contrast, the SF indicator leverages the time-varying characteristics of the score function $S(t)$ to measure the irreversibility of processes. The magnitude of $S(t)$ corresponds to the degree of deviation from equilibrium; a higher value signifies a more marked deviation. This indicator does not require prior models and can be estimated directly from observed time series (e.g., climate, ecological, or biomedical signals), making it suitable for nonlinear and non-stationary systems. Experiments (see Chapter 4) show that SF indicators often exhibit sharp peaks or sustained upward trends in critical regions, accompanied by significant changes in statistical properties such as variance and autocorrelation. The two indicators are complementary in methodology; the OM indicator applies to systems with clear dynamical models and has a well-defined mathematical-physical foundation, while the SF indicator is more suitable for complex systems with abundant data but unclear mechanisms, offering broader applicability and computational feasibility.

\section{Experiment}

\subsection{Data Processing and Dimensionality Reduction}
\hspace*{1em}The Electroencephalography (EEG) data processing pipeline employed in this study is as follows. The raw data consist of multi-channel matrices, sampled at 256 Hz during pre-ictal and ictal periods. First, amplitude normalization was applied to scale the data to the range $[-0.5, 0.5]$. To reduce computational load, downsampling was subsequently performed by selecting one point from every 16 consecutive points, resulting in an effective sampling interval of 0.0625 s. Furthermore, based on medical annotations, $T = 500$ was defined as the critical time point for seizure onset.

The results demonstrate that although several dimensionality reduction methods can distinguish pre- and post-critical transition data in epileptic seizures, we ultimately selected OLPP, DM, Isomap, PCA, Kernel PCA, and NPE for constructing early-warning indicators (see Table~\ref{tab: dim-score}(a,d)). This selection is based on a comprehensive evaluation where these six methods collectively exhibited superior sensitivity, robustness, and interpretability in capturing critical transitions. Specifically, all six methods consistently detected a distribution shift around time point $T \approx 600$, yet with distinct behavioral patterns: For instance, kernel PCA exhibited significant mean shifts and increased variance after the transition, whereas DM, Isomap, OLPP, and NPE maintained stable mean values but displayed markedly elevated variance post-transition. Overall, OLPP, DM, Isomap, PCA, Kernel PCA, and NPE not only provided clearer separation between pre- and post-seizure states but also offered a balanced trade-off between preserving local structures and capturing global manifold geometry, making them particularly suitable for the subsequent derivation of probabilistic early-warning measures.

\subsection{Learning SDE in Latent Space}

\hspace*{1em}Our objective is to solve the Schrödinger bridge problem. Specifically, we aim to find effective parametric representations for the diffusion and drift terms of its associated stochastic differential equation. To this end, we employ neural networks to approximate the equation and learn the score function via score matching, ultimately achieving the construction of the Schrödinger bridge.

By employing polynomial basis functions for fitting, with the time step set as $\Delta t = 0.0625$, the entire time domain is discretized into 1500 time points, thereby generating 1499 data pairs of consecutive time points $\{(z(t_{i+1}), z(t_i))\}$. Among these, 1200 pairs are used for training, while the remaining 299 pairs are reserved for validation. Through this fitting process, we obtain the following SDE

\begin{equation}
    dX = f(x,t)dt + g(x,t) dW_t
\end{equation}
This article focuses on the SDE corresponding to anisotropic diffusion maps, which perform better in dimensionality reduction. The SDE learned from various manifold learning dimensionality reduction methods is shown in the following Table \ref{tab: SDE learn}.

\begin{table}[htbp]
\centering
\caption{SDE learned by various manifold learning dimensionality reduction methods}
\label{tab: SDE learn}
\renewcommand{\arraystretch}{1.5}
\begin{tabular}{ccc}
\toprule
\textbf{Result/Method} & \textbf{Latent SDE} & \textbf{Evaluation} \\
\midrule

\begin{minipage}[c]{2.5cm}
\centering
\textbf{Isomap}
\end{minipage}
& 
\begin{minipage}[c]{5cm}
\centering
$\displaystyle d{X}_t = (-0.0121+0.4773x+0.5140x^2-1.5114x^3)\mathrm{d}t -0.5023dW_t$
\end{minipage}
& 
\begin{minipage}[c]{5cm}
\centering
\includegraphics[width=0.85\linewidth, height=2cm]{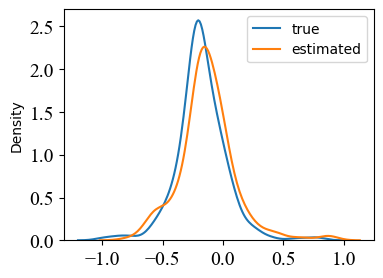}
\end{minipage} \\

\addlinespace

\begin{minipage}[c]{2.5cm}
\centering
\textbf{Kernel PCA}
\end{minipage}
& 
\begin{minipage}[c]{5cm}
\centering
$\displaystyle d{X}_t = (0.2160+0.4978x-0.5168x^2-1.4849x^3)\mathrm{d}t -0.4264dW_t$
\end{minipage}
& 
\begin{minipage}[c]{5cm}
\centering
\includegraphics[width=0.85\linewidth, height=2cm]{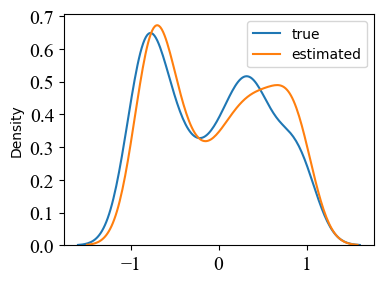}
\end{minipage} \\

\addlinespace

\begin{minipage}[c]{2.5cm}
\centering
\textbf{NPE}
\end{minipage}
& 
\begin{minipage}[c]{5cm}
\centering
$\displaystyle d{X}_t = (0.0022+0.4823x-0.5070x^2-1.5086^3)\mathrm{d}t -0.5006dW_t$
\end{minipage}
& 
\begin{minipage}[c]{5cm}
\centering
\includegraphics[width=0.85\linewidth, height=2cm]{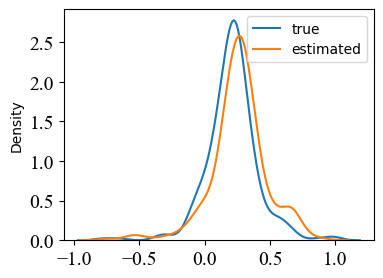}
\end{minipage} \\

\addlinespace

\begin{minipage}[c]{2.5cm}
\centering
\textbf{OLPP}
\end{minipage}
& 
\begin{minipage}[c]{5cm}
\centering
$\displaystyle d{X}_t = (-0.0163+0.4810x+0.5038x^2-1.5082^3)\mathrm{d}t -0.4897dW_t$
\end{minipage}
& 
\begin{minipage}[c]{5cm}
\centering
\includegraphics[width=0.85\linewidth, height=2cm]{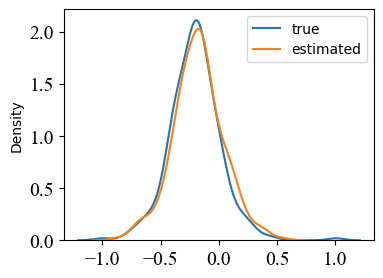}
\end{minipage} \\

\addlinespace

\begin{minipage}[c]{2.5cm}
\centering
\textbf{DM}
\end{minipage}
& 
\begin{minipage}[c]{5cm}
\centering
$\displaystyle d{X}_t = (-0.0650+0.4746x-0.5048x^2-1.5105x^3)\mathrm{d}t -0.4711dW_t$
\end{minipage}
& 
\begin{minipage}[c]{5cm}
\centering
\includegraphics[width=0.85\linewidth, height=2cm]{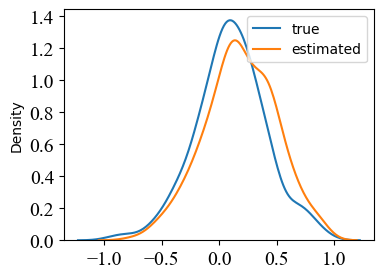}
\end{minipage} \\

\addlinespace

\begin{minipage}[c]{2.5cm}
\centering
\textbf{PCA}
\end{minipage}
& 
\begin{minipage}[c]{5cm}
\centering
$\displaystyle d{X}_t = (0.0026+0.4809x+0.5067x^2-1.5101x^3)\mathrm{d}t -0.5026dW_t$
\end{minipage}
& 
\begin{minipage}[c]{5cm}
\centering
\includegraphics[width=0.85\linewidth, height=2cm]{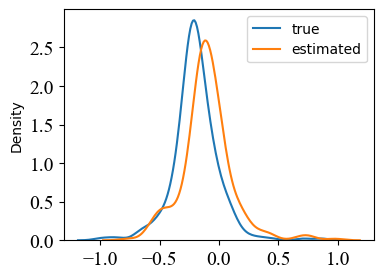}
\end{minipage} \\

\bottomrule
\end{tabular}

\vspace{0.2cm}
\footnotesize 
\caption*{The probability density comparison between the true distribution and the model-predicted distribution is shown in the table. Based on 1499 consecutive data pairs generated with $dt = 0.0625$ (1200 for training and 299 for validation), two separate single-layer networks were employed to approximate the drift term (using the polynomial basis [1, x, x², x³]) and the diffusion term (using a constant basis with Softplus activation) of an SDE. The model was trained for 100 epochs using the AdamW optimizer (learning rate = 0.005) with a batch size of 1200 and a negative log-likelihood loss function, ultimately achieving close agreement between the true and predicted distributions.}
\end{table}

The consistency between our approximate and true data densities is evident (as shown in Table \ref{tab: dim-score}). Among them, Kernel PCA exhibits two metastable states in the near-sighted view of SDE under the projected coordinates, while OLLP, Isomap, NPE, PCA, and DM display a single peak in the near-sighted view under the projected coordinates.

\subsection{Learning Score Function via Score Matching}
The stochastic differential equation corresponding to the Schrödinger bridge was solved using a direct sampling method based on score matching, applied to the fitted SDE. This approach utilizes a 3-layer fully connected neural network to approximate the score function. The network's input and output dimensions are consistent with the dimensionality of the SDE, with the hidden layer dimension uniformly set to 218. Training was performed using the Adam optimizer with a learning rate of 1e-3. Through this method, we ultimately obtained the following Schrödinger bridge equation

\begin{equation}
    dX = (f(x,t) - g^2(x,t)S(x,t))dt + g(x,t)dW_t,
\end{equation}

\noindent where $S$ denotes the score function, defined as the gradient with respect to $x$ of the logarithm of the probability density $p$, i.e., $S = \nabla_x \log p(x)$.

To evaluate the accuracy of the learned score function, we trained the network using 1000 SDE trajectories, resulting in a score function of dimension (1500, 1000). Table \ref{tab: dim-score} (b,e) presents a comparison between the learned score function and the true score function at four specific time points. The results demonstrate close agreement between the two, indicating that our model effectively approximates the true score function. The learned score function based on our projection coordinates and the corresponding projection coordinates are shown in Table \ref{tab: dim-score} (c,f).

\begin{table}[htbp]
\centering
\caption{$\phi_1$ and the Score Function Across Manifold Learning Methods}
\label{tab: dim-score}
\begin{tabular}{lccc}
\toprule
\textbf{Result/Method} & \textbf{OLPP} & \textbf{DM} & \textbf{Isomap} \\
\midrule
(a)$\phi_1$ Data & 
\includegraphics[width=0.23\linewidth]{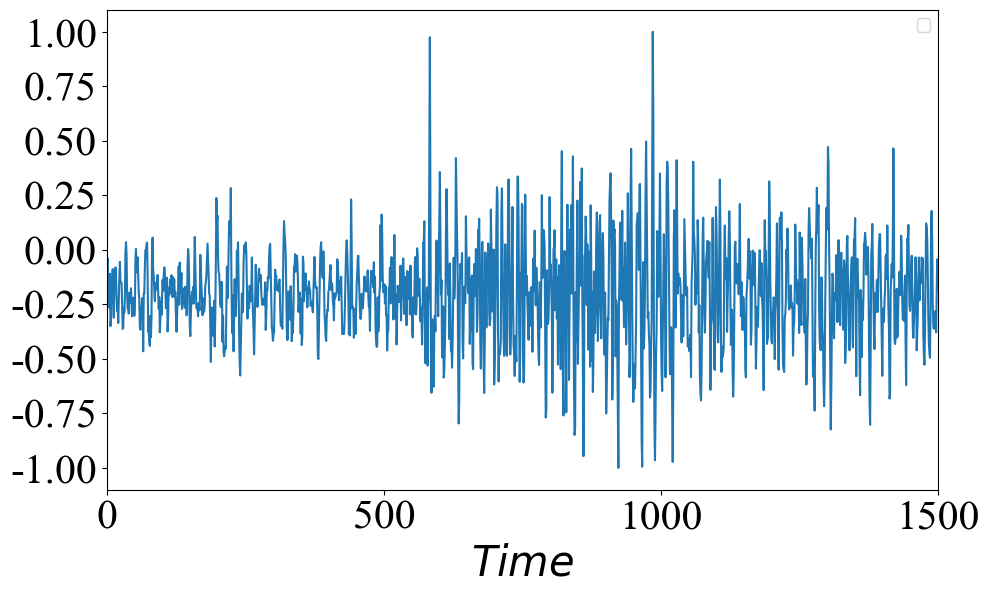} &
\includegraphics[width=0.23\linewidth]{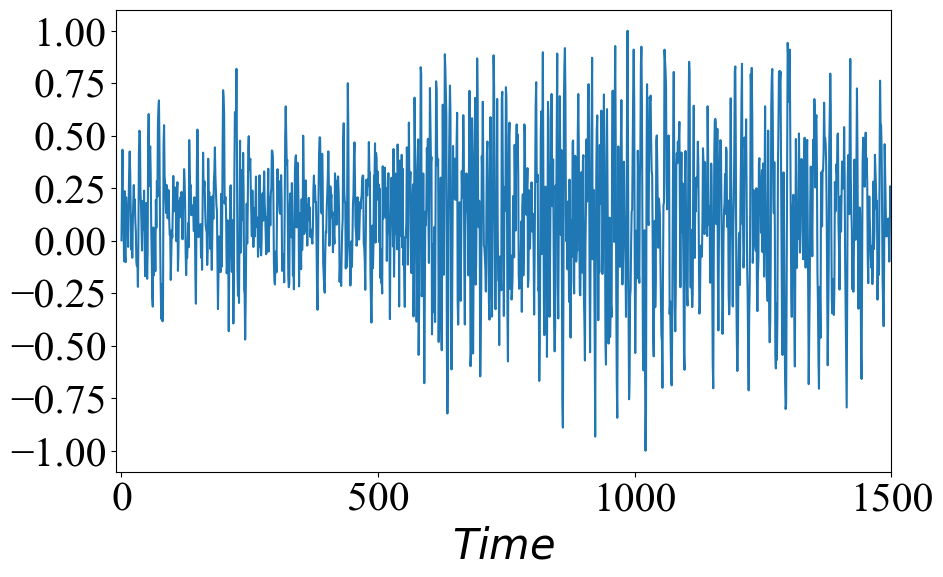} &
\includegraphics[width=0.23\linewidth]{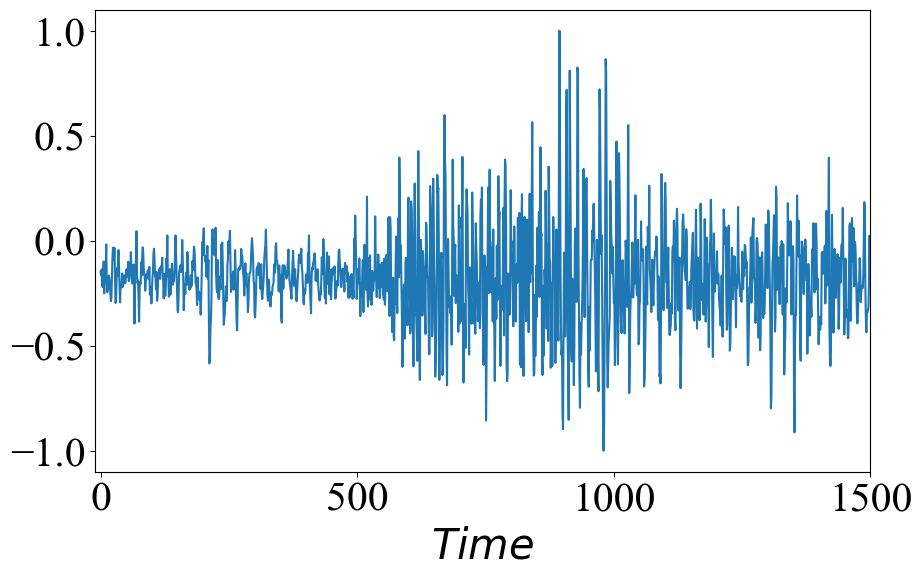} \\
\addlinespace[3pt]
(b)Ture vs Learn & 
\includegraphics[width=0.23\linewidth]{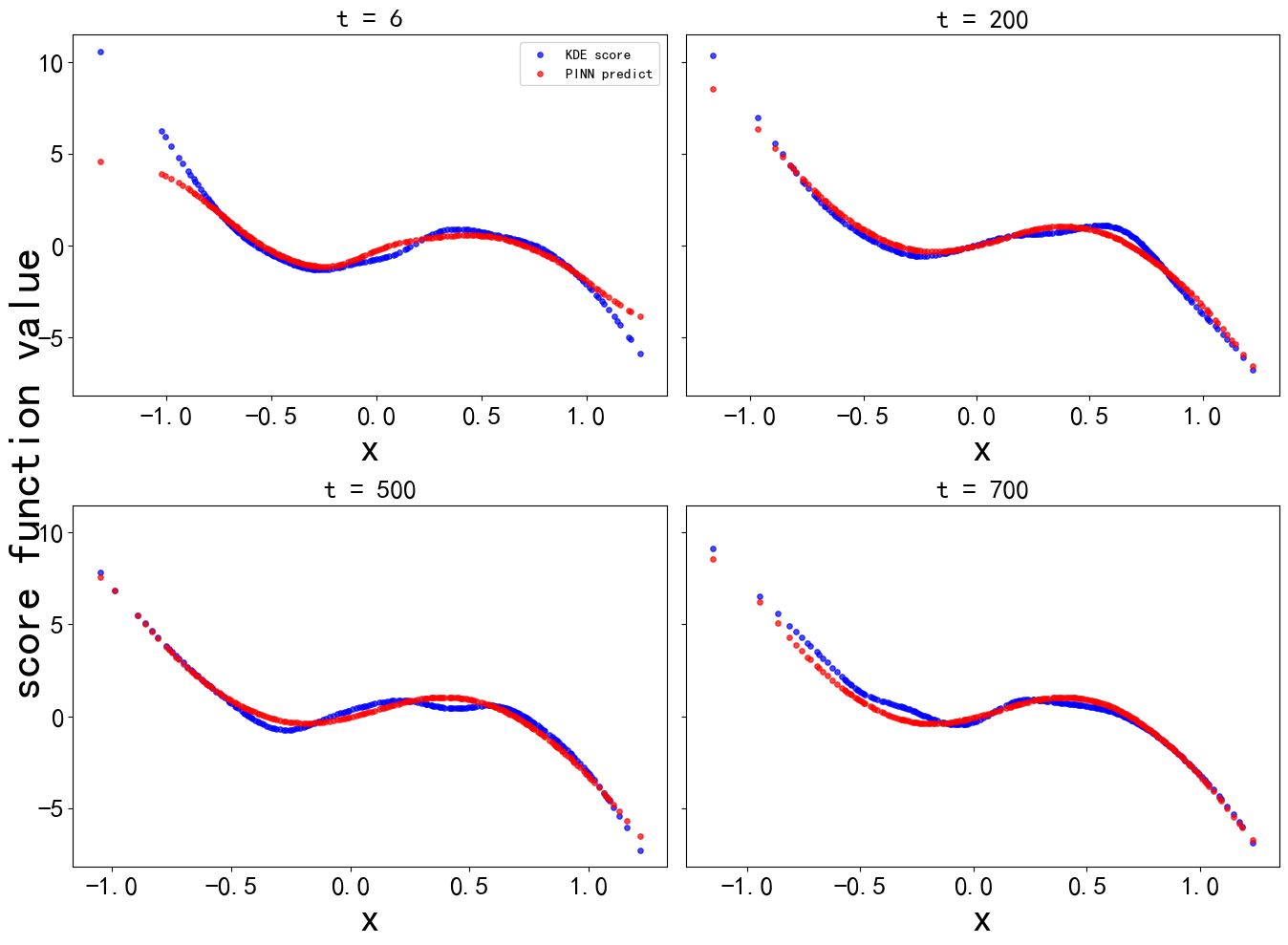} &
\includegraphics[width=0.23\linewidth]{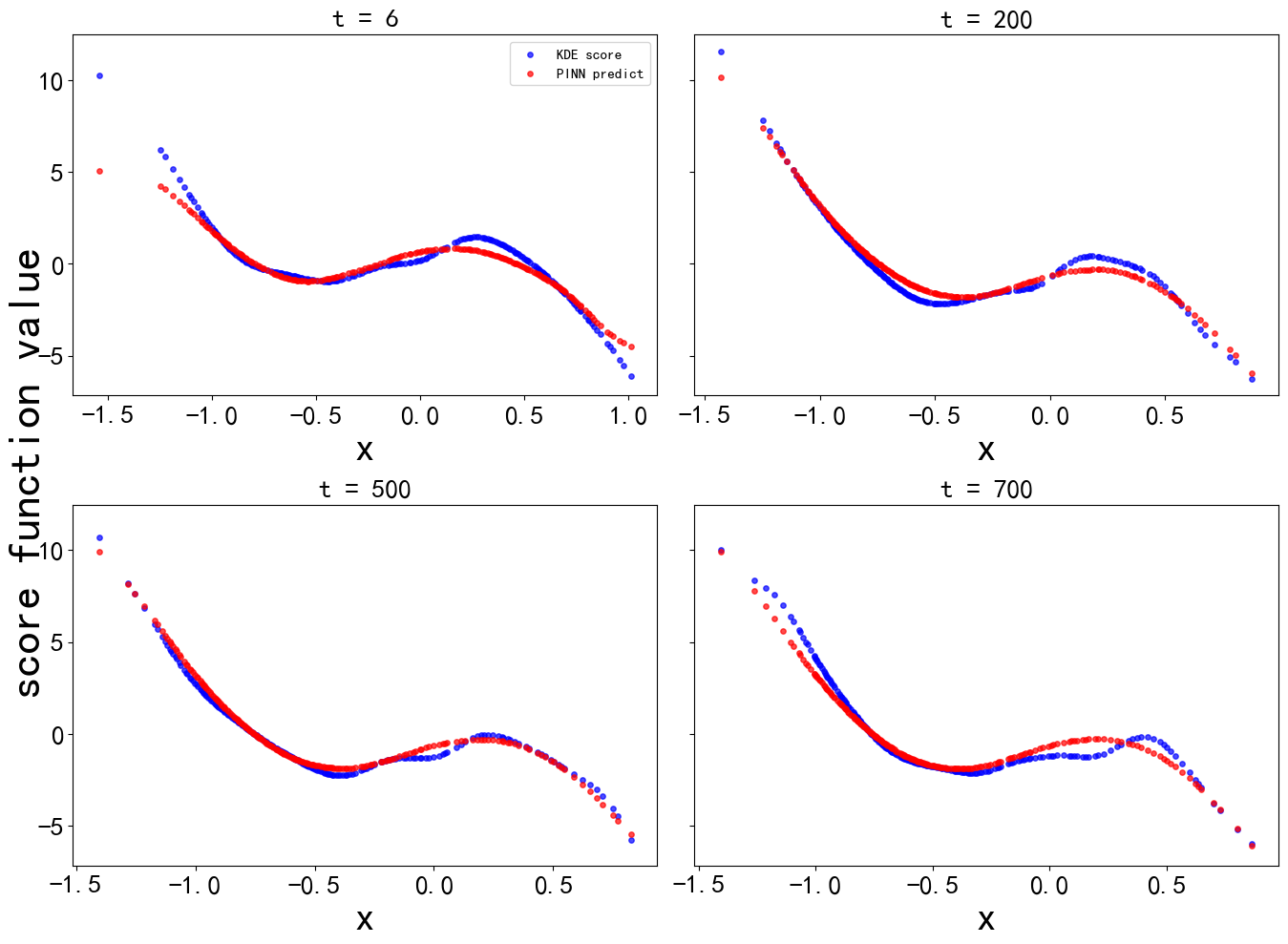} &
\includegraphics[width=0.23\linewidth]{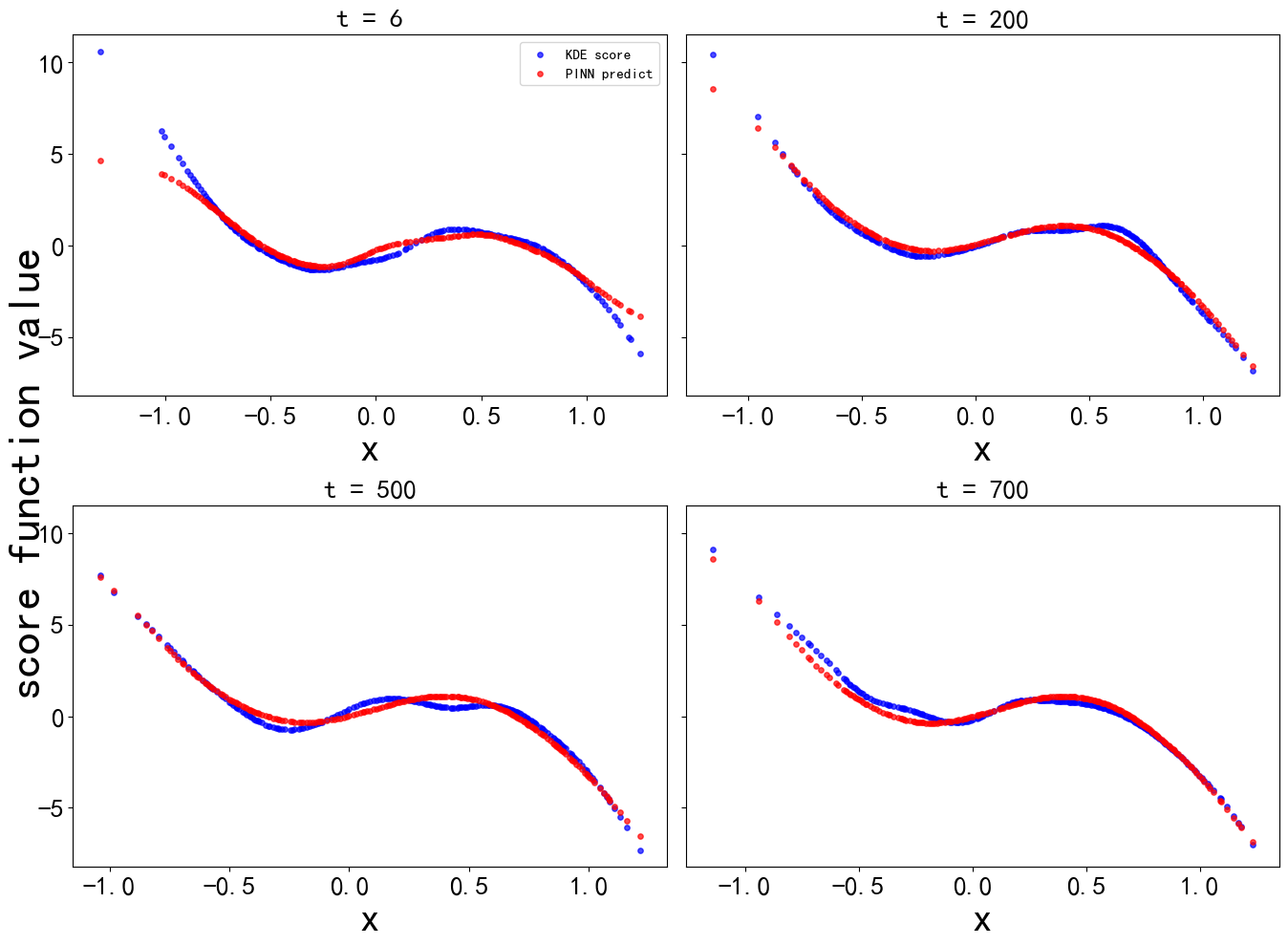} \\
\addlinespace[3pt]
(c)Score Function & 
\includegraphics[width=0.23\linewidth]{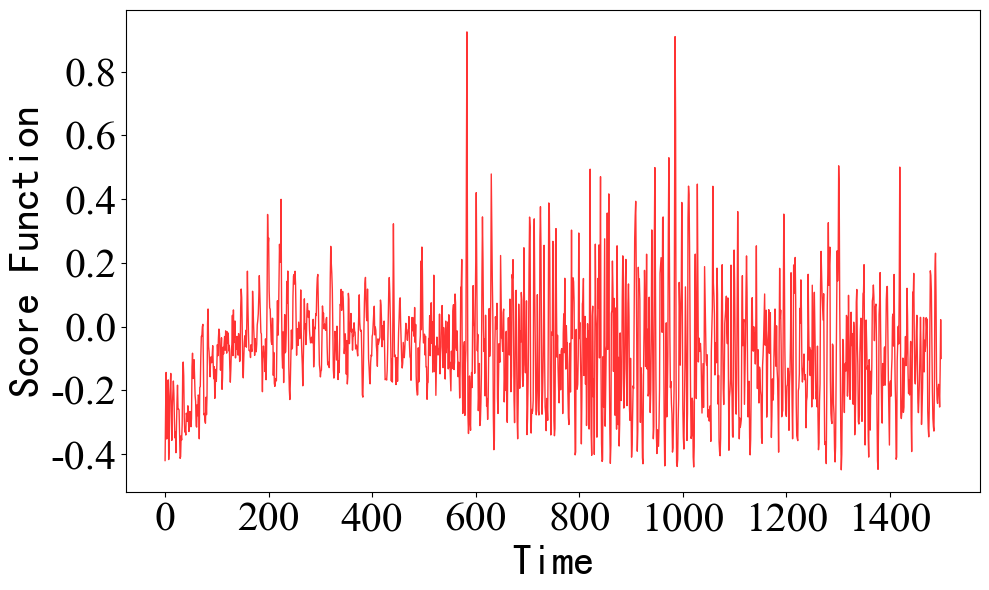} &
\includegraphics[width=0.23\linewidth]{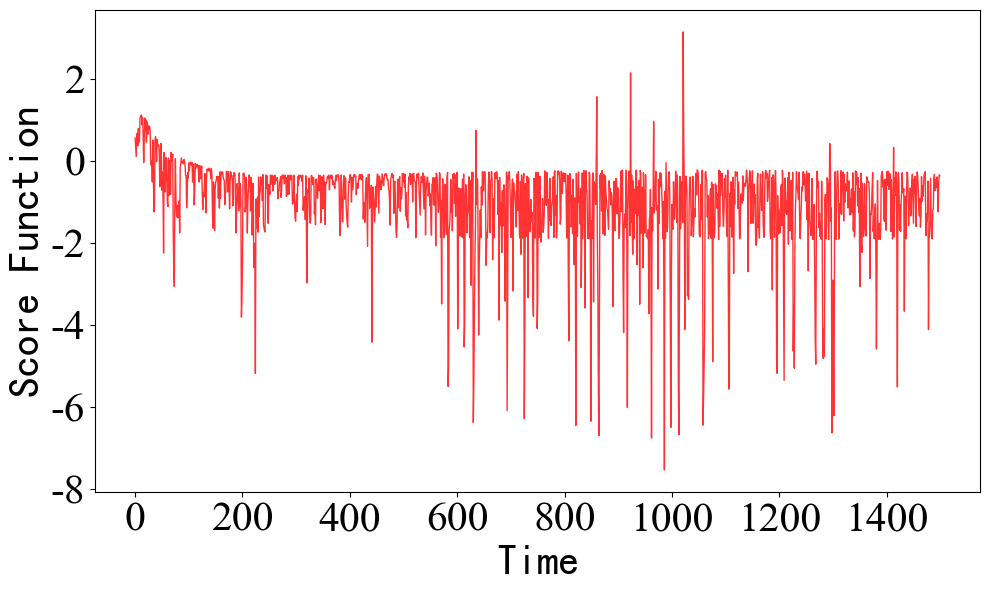} &
\includegraphics[width=0.23\linewidth]{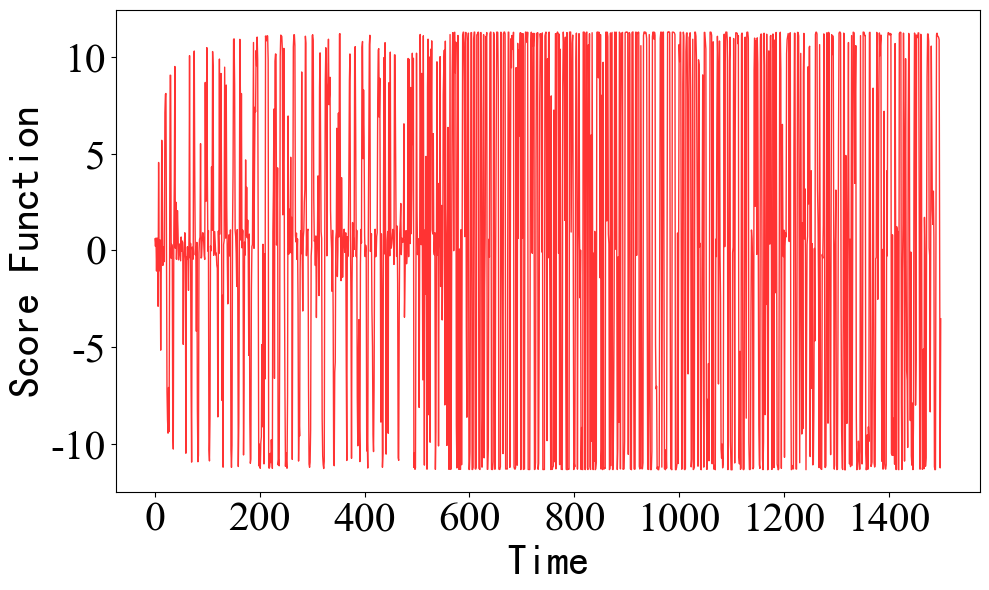} \\
\midrule
\textbf{Result/Method} & \textbf{NPE} & \textbf{PCA} & \textbf{Kernel PCA} \\
\midrule
(d)$\phi_1$ Data & 
\includegraphics[width=0.23\linewidth]{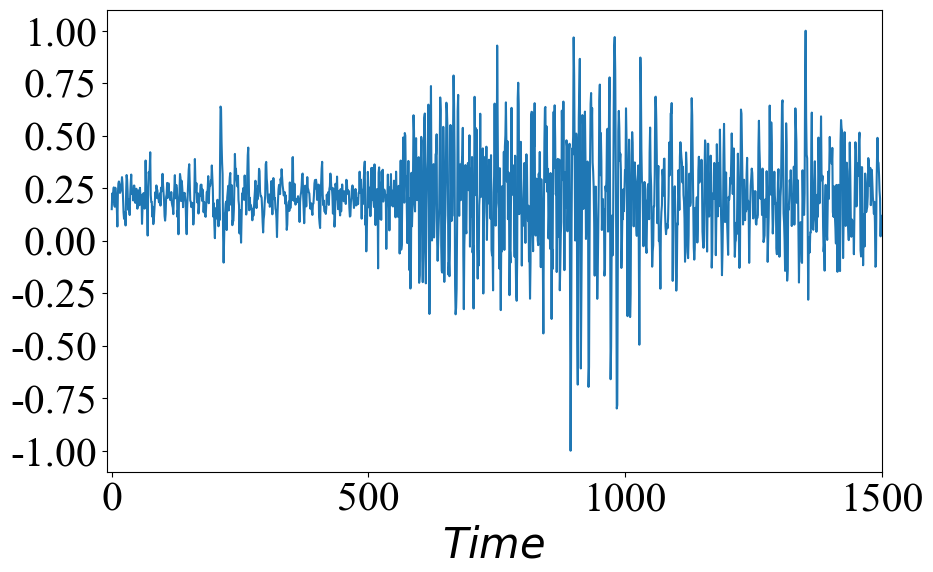} &
\includegraphics[width=0.23\linewidth]{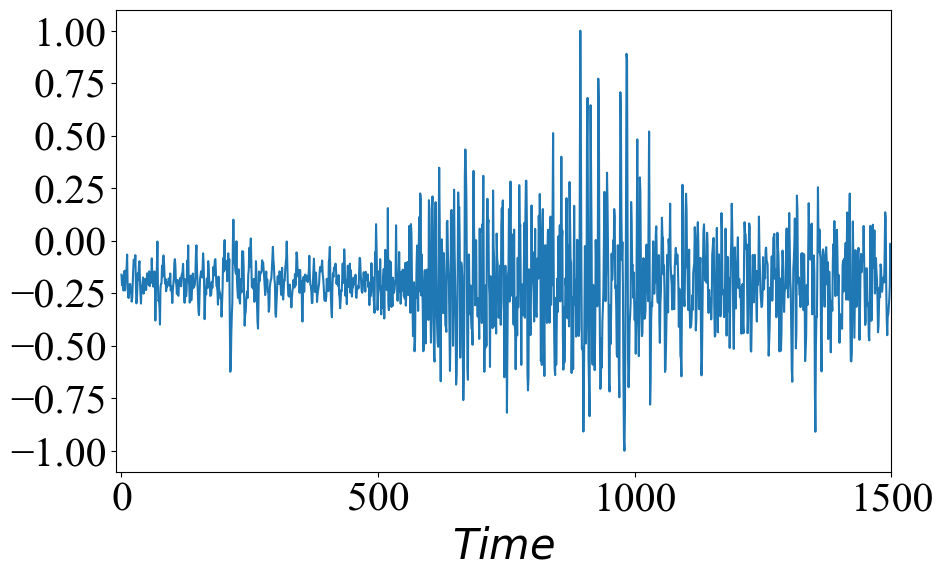} &
\includegraphics[width=0.23\linewidth]{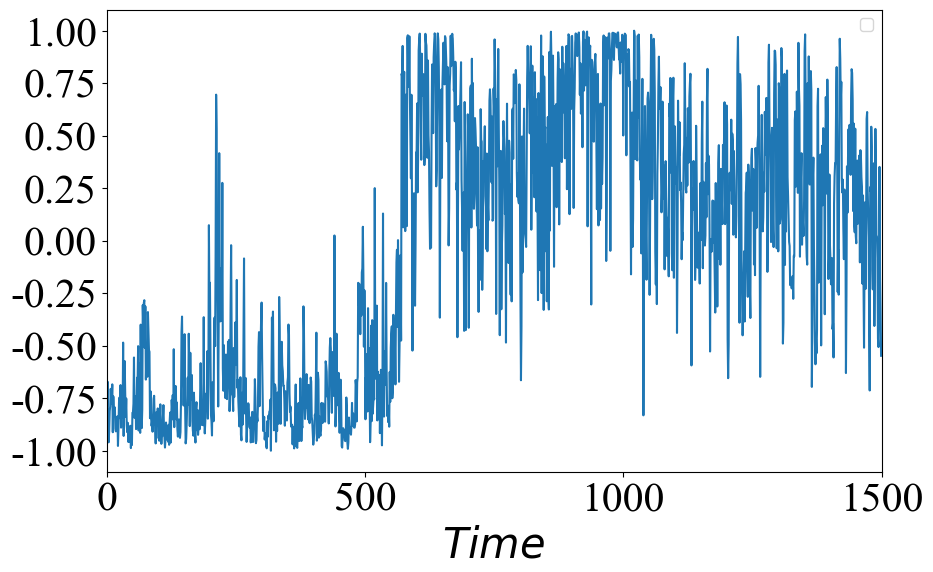} \\
\addlinespace[3pt]
(e)Ture vs Learn& 
\includegraphics[width=0.23\linewidth]{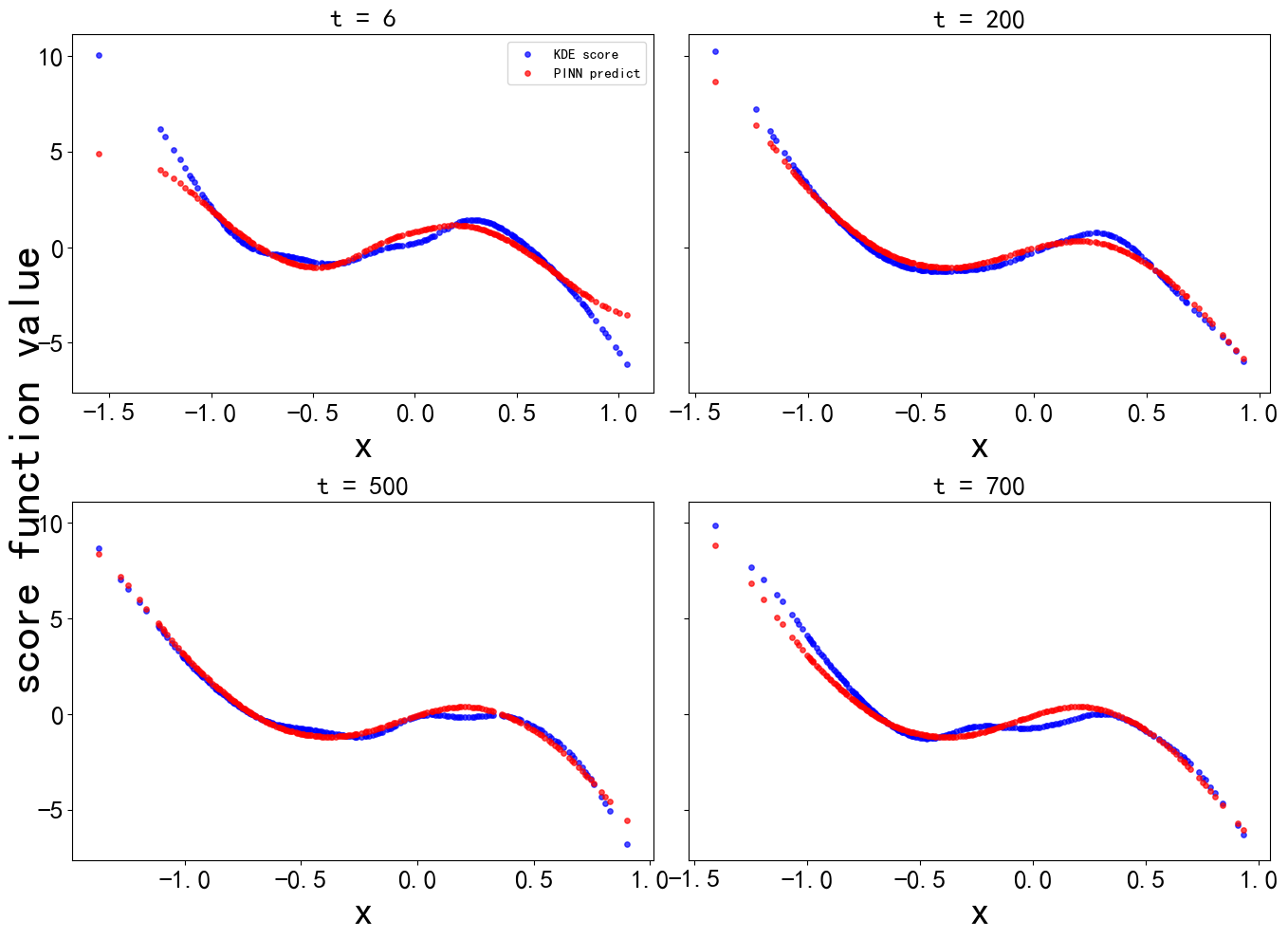} &
\includegraphics[width=0.23\linewidth]{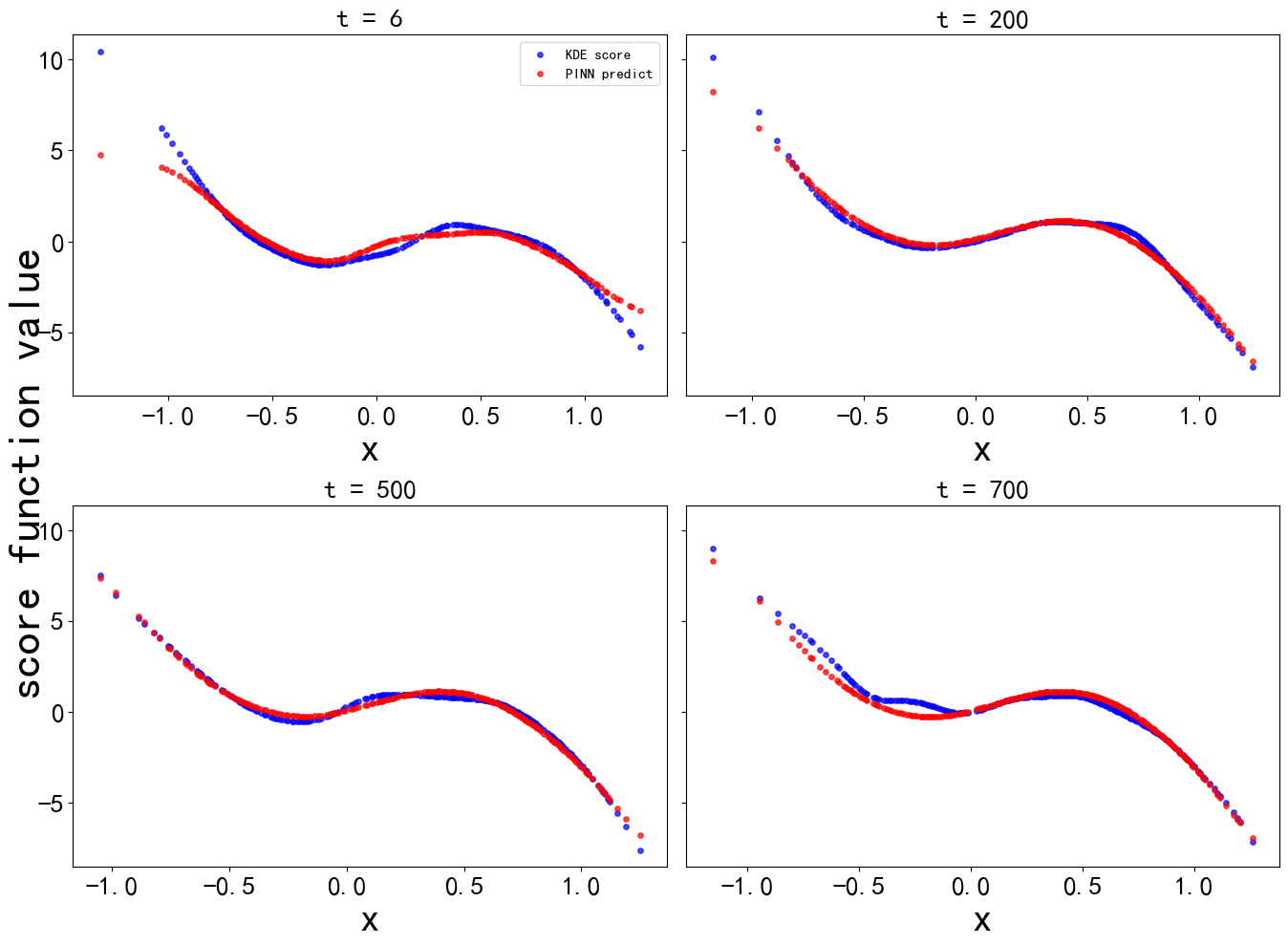} &
\includegraphics[width=0.23\linewidth]{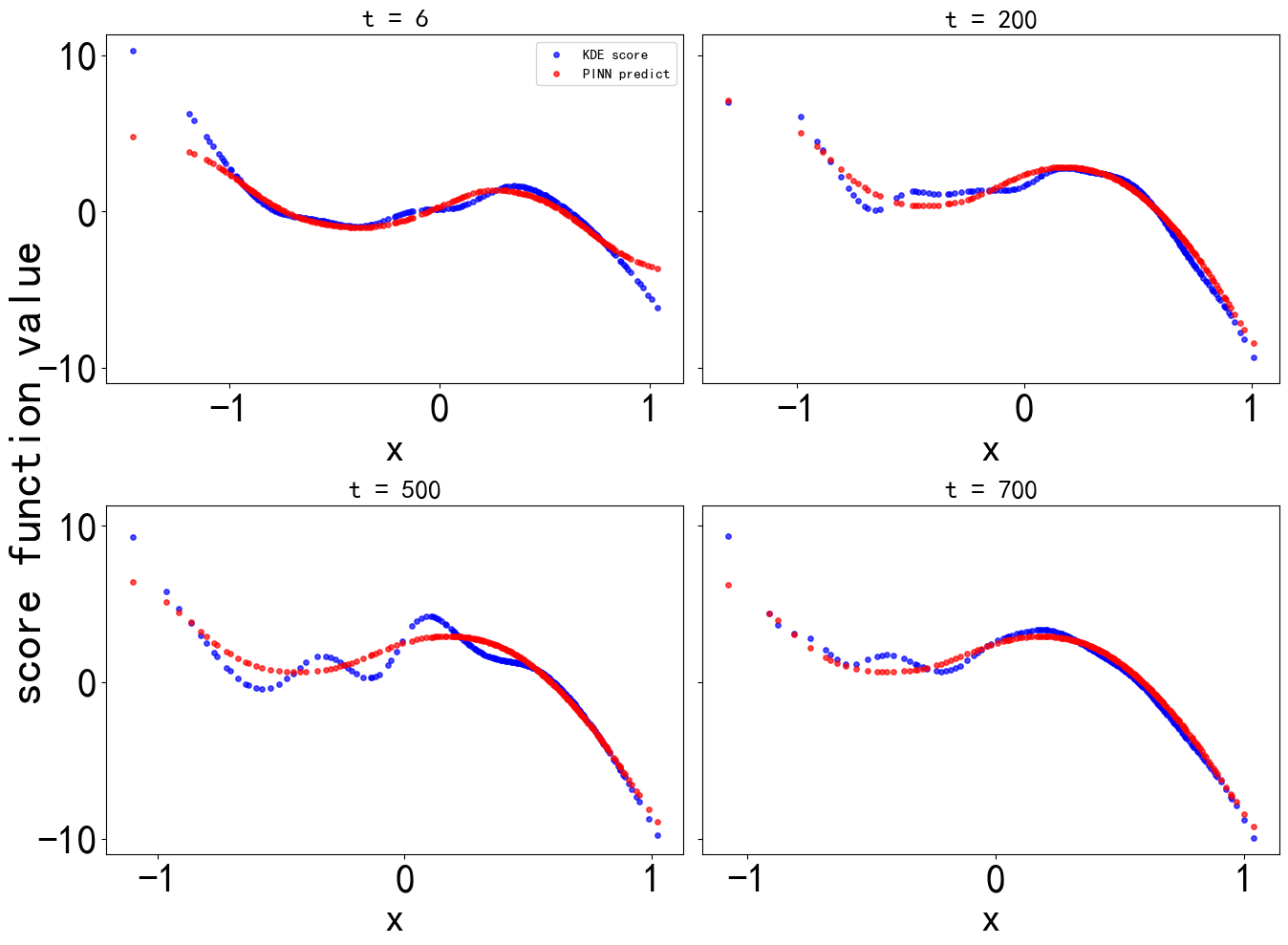} \\
\addlinespace[3pt]
(f)Score Function & 
\includegraphics[width=0.23\linewidth]{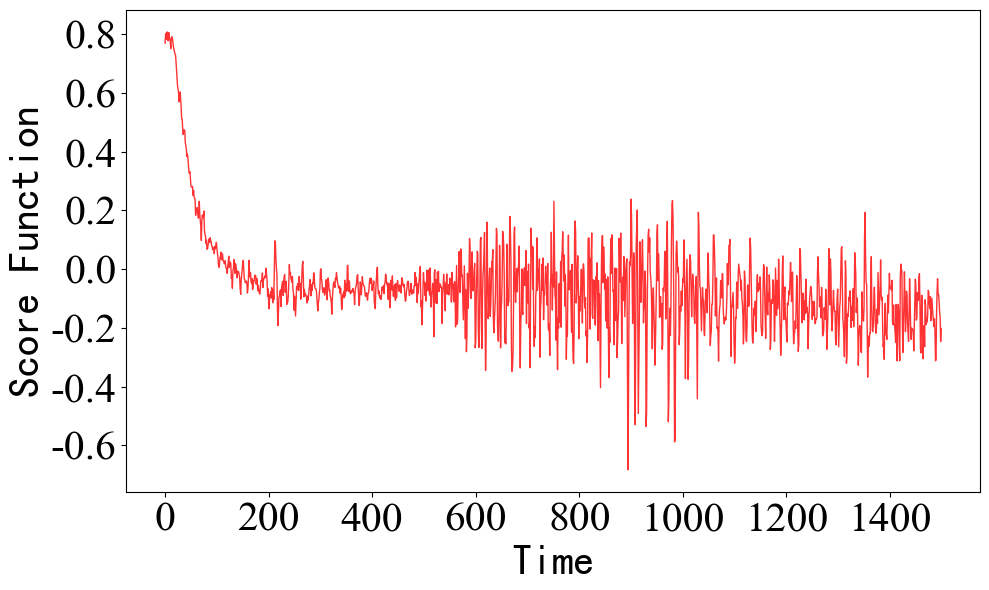} &
\includegraphics[width=0.23\linewidth]{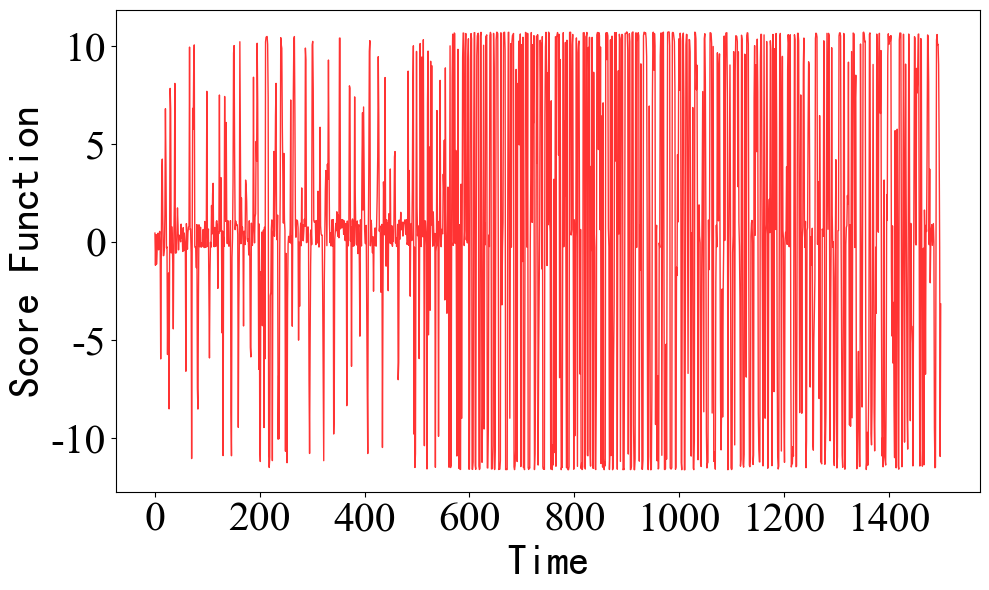} &
\includegraphics[width=0.23\linewidth]{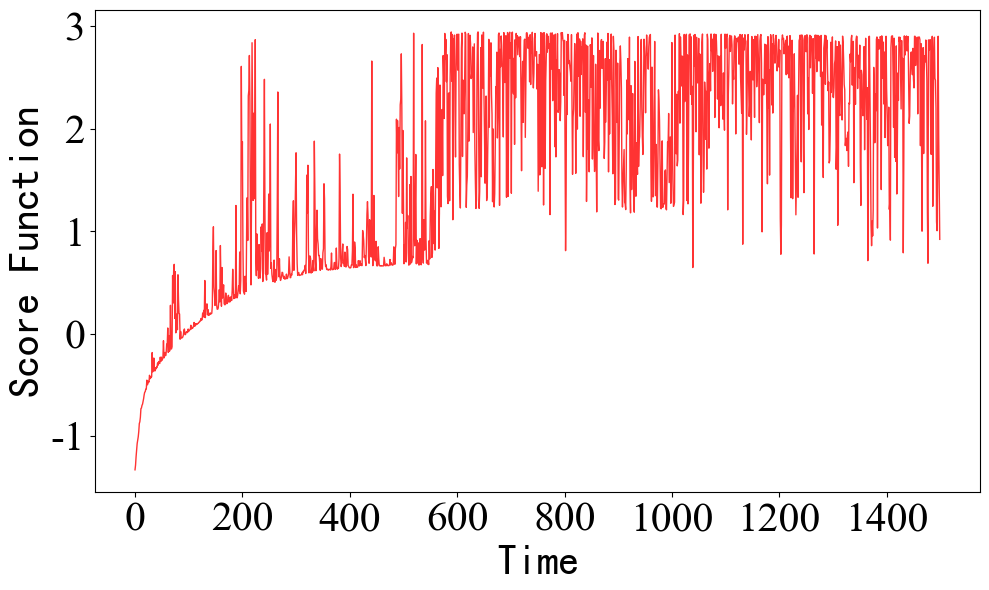} \\
\bottomrule
% 关键修改：用\multicolumn跨4列，p{0.95\textwidth}保证文本自动换行且不贴边
\multicolumn{4}{p{0.95\textwidth}}{%
  \textbf{} (a,d) The one-dimensional $\phi_1$ data is obtained by applying manifold learning methods to reduce the dimensionality of the data. (b,e) The score function S(x) learned by the trained Score Model (SM) (red) is quantitatively compared with the score function computed via kernel density estimation (blue) at four specific time points (t = 6, 200, 500, 700). (c,f) The temporal evolution of the score function S(t) learned by the SM is fully illustrated. This method employs a three-layer fully connected neural network to approximate the score function of a stochastic differential equation (SDE), with the hidden layer dimension fixed at 218. Training is performed using 1,000 sampled trajectories and the Adam optimizer with a learning rate of 1e-3 and a step size $dt = 0.0625$.
} \\
\end{tabular}
\end{table}

\subsection{Early Warning Indicator}

\hspace*{1em} This study is based on a time series comprising 1500 data points. We employ a sliding window approach to calculate the variation of each indicator over time and compare it with the standard deviation (Std). The window length is set to \(l=300\), and the sliding step is \(\Delta t=0.0625\). For the Onsager-Machlup (OM) functional indicator, we define it as the ratio of two OM functionals (Eq. (11)) in each window, i.e.,
\begin{equation}
    \mathrm{OM}(l,t) / \mathrm{OM}(0.9l, t-0.9l).
\end{equation}
Meanwhile, we propose a new early warning indicator SF, which is defined as the integral of the square of the score function averaged over the state variable \(x\) (Eq. (13)). To maintain consistency in the calculation, the SF indicator also adopts the ratio form, i.e., 
\begin{equation}
    \mathrm{SF}(l,t) / \mathrm{SF}(0.9l, t-0.9l).
\end{equation}

As shown in Table~\ref{tab: Indicator}(b,e), the standard deviation (Std) curves (green curves) derived from the four dimensionality reduction methods all exhibit interpretable patterns. In the dimensionality-reduced data, both the OM and SF indicators issue warning signals prior to the sharp rise in standard deviation. Specifically, the warning signals based on manifold learning show relatively gradual changes and higher noise levels around \(T \approx 600\).

Compared to the OM indicator, the SF indicator demonstrates higher sensitivity and robustness. For instance, as shown for OLPP in Table~\ref{tab: Indicator}(c), the SF indicator not only accurately identifies the critical point at \(T \approx 600\) but is also highly responsive to dynamic characteristics at later time points such as \(T \approx 1000\) and \(T \approx 1300\). Although the signal patterns vary across different dimensionality reduction methods, the SF indicator consistently captures key transitions across multiple approaches.

While the computation of the OM indicator depends on solving stochastic differential equations, the SF indicator is derived directly from the predicted score function. Furthermore, compared to the OM indicator, the SF indicator provides earlier and more pronounced warnings of critical states and more reliably identifies the dynamic transition characteristics across different phases before and after epileptic seizures.

\begin{table}[htbp]
\centering
\caption{OM and SF Indicator Across Manifold Learning Methods}
\label{tab: Indicator}
% 表格是lccc共4列：l(左对齐) + 3个c(居中)
\begin{tabular}{lccc}
\toprule
\textbf{Result/Method} & \textbf{OLPP} & \textbf{DM} & \textbf{Isomap} \\
\midrule
(a)$\phi_1$ Data & 
\includegraphics[width=0.23\linewidth]{fig2/phi1_OLPP.png} &
\includegraphics[width=0.23\linewidth]{fig2/phi1_DM.png} &
\includegraphics[width=0.23\linewidth]{fig2/phi1_Isomap.png} \\
\addlinespace[3pt]
(b)OM Indicator & 
\includegraphics[width=0.23\linewidth]{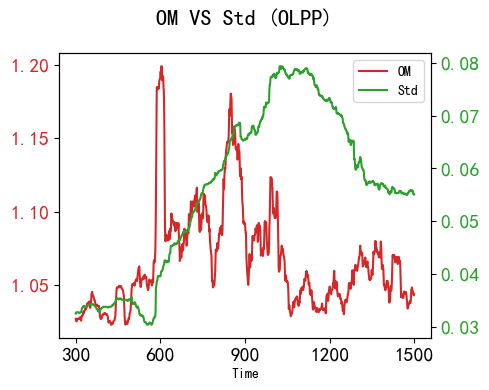} &
\includegraphics[width=0.23\linewidth]{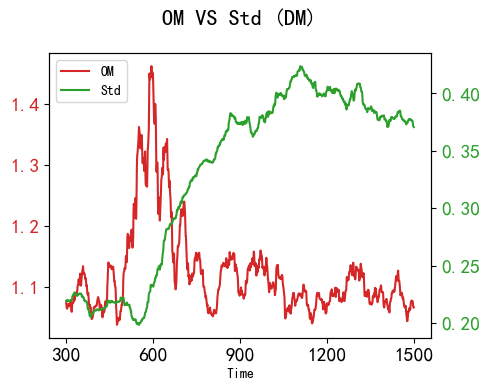} &
\includegraphics[width=0.23\linewidth]{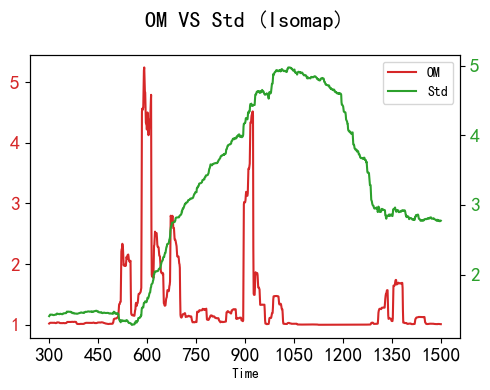} \\
\addlinespace[3pt]
(c)SF Indicator & 
\includegraphics[width=0.23\linewidth]{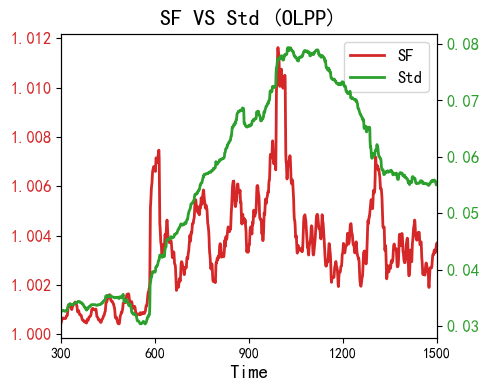} &
\includegraphics[width=0.23\linewidth]{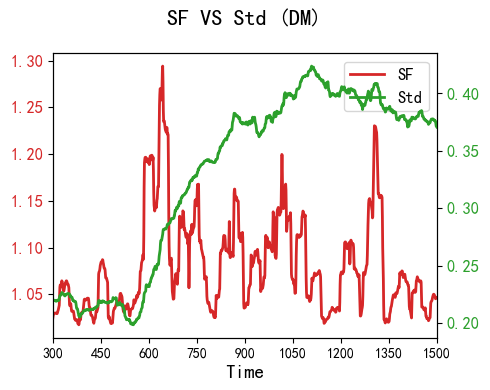} &
\includegraphics[width=0.23\linewidth]{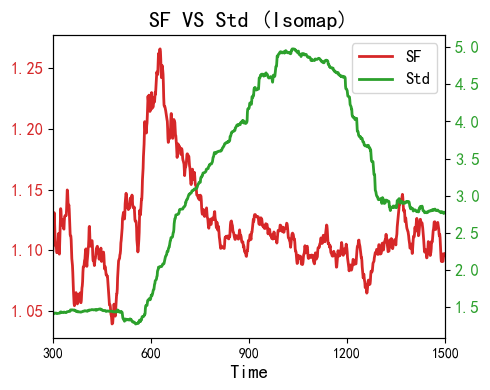} \\
\midrule
\textbf{Result/Method} & \textbf{NPE} & \textbf{PCA} & \textbf{Kernel PCA} \\
\midrule
(d)$\phi_1$ Data& 
\includegraphics[width=0.23\linewidth]{fig2/phi1_NPE.png} &
\includegraphics[width=0.23\linewidth]{fig2/phi1_PCA.png} &
\includegraphics[width=0.23\linewidth]{fig2/phi1_Kernel_PCA.png} \\
\addlinespace[3pt]
(e)OM Indicator & 
\includegraphics[width=0.23\linewidth]{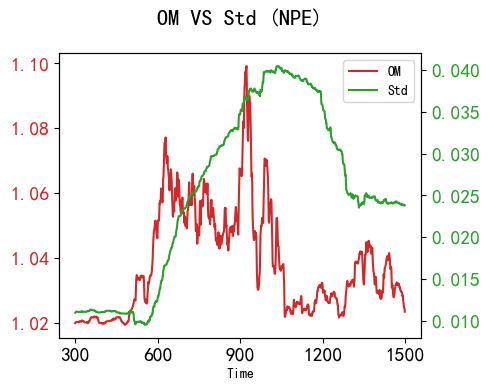} &
\includegraphics[width=0.23\linewidth]{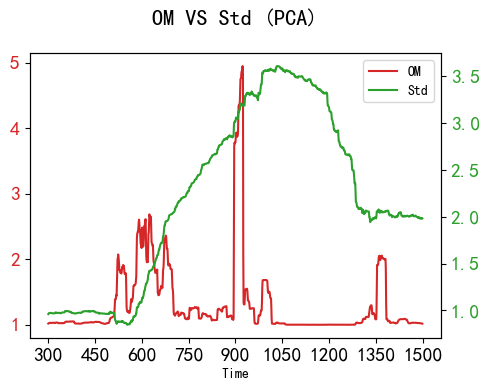} &
\includegraphics[width=0.23\linewidth]{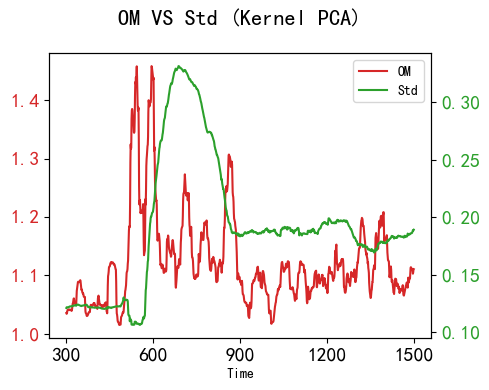} \\
\addlinespace[3pt]
(f)SF Indicator & 
\includegraphics[width=0.23\linewidth]{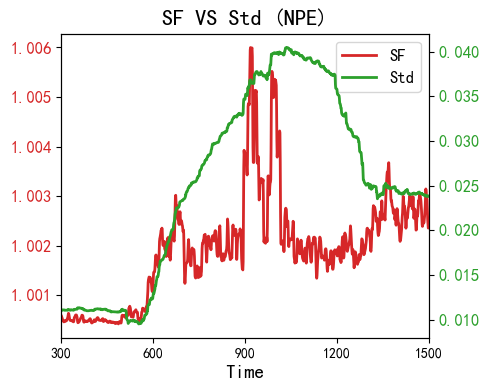} &
\includegraphics[width=0.23\linewidth]{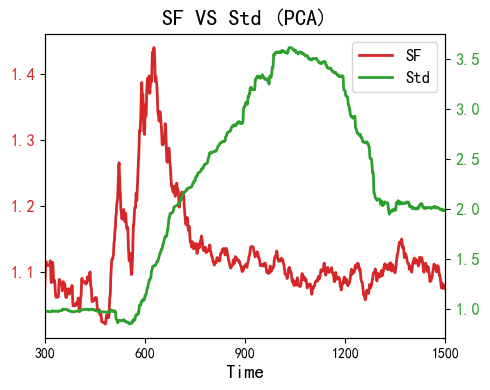} &
\includegraphics[width=0.23\linewidth]{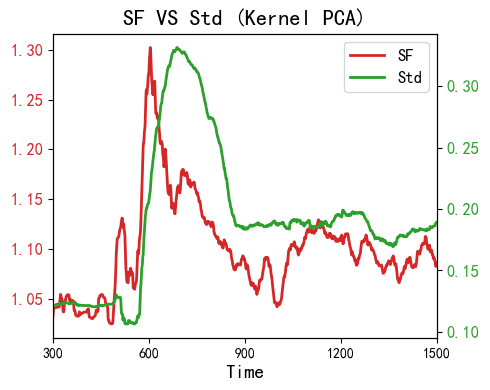}\\
\bottomrule
\multicolumn{4}{p{0.95\textwidth}}{%
  \textbf{} Based on a time-series comprising 1500 data points, a sliding-window analysis was employed (window length l = 300, step size $dt = 0.0625$). (a, d) Principal variable $\phi_1$ obtained after dimensionality reduction via different methods. (b, e) Comparison between the OM early-warning indicator (red, calculated as Eq. (18)) and the standard deviation (Std) of $\phi_1$ (green). (c, f) Comparison between the SF early-warning indicator (red, calculated as Eq. (19)) and the Std of $\phi_1$ (green). All indicators are presented in ratio form. The time axis ranges from 300 to 1500.
} \\
\end{tabular}
\end{table}

\subsection{Validation}

\hspace*{1em}To validate the generalization performance of the proposed epilepsy prediction method, we tested it on a new EEG dataset. These data were collected from epilepsy patients during pre-ictal and ictal phases at a sampling frequency of 256 Hz. Based on medical annotations, \( T = 780 \) was defined as the critical time point for seizure onset.

The raw data underwent standard preprocessing prior to analysis. First, amplitude normalization scaled the signals to the range \([-0.5, 0.5]\) to mitigate inter-individual amplitude variations. Subsequently, downsampling was performed by selecting one point from every 16 consecutive points. This step preserved critical dynamic information while reducing the effective sampling interval to 0.0625 seconds, thereby enhancing computational efficiency for subsequent steps.

The preprocessed data were then separately input into the six manifold learning-based dimensionality reduction models. These models effectively capture essential changes in brain network connectivity and dynamics associated with epileptic seizures. As shown in Table \ref{tab: new data indicator}(a,d), the dimensionally-reduced data clearly reveal the evolutionary trajectory from the pre-ictal to the ictal state on a low-dimensional manifold, confirming that this feature extraction method can effectively distinguish between different brain states.

Subsequently, we applied the trained score-matching model to predict the score function and employed a sliding-window technique to dynamically compute early-warning indicators within the Schrödinger bridge framework, enabling online and continuous seizure monitoring. The resulting time-series plots of the OM and SF indicators are presented in Table \ref{tab: new data indicator}.

\begin{table}[htbp]
\centering
\caption{Indicator Results on New Data Across Manifold Learning Methods}
\label{tab: new data indicator}
% 表格是lccc共4列：l(左对齐) + 3个c(居中)
\begin{tabular}{lccc}
\toprule
\textbf{Result/Method} & \textbf{OLPP} & \textbf{DM} & \textbf{Isomap} \\
\midrule
(a)$\phi_2$ Data & 
\includegraphics[width=0.23\linewidth]{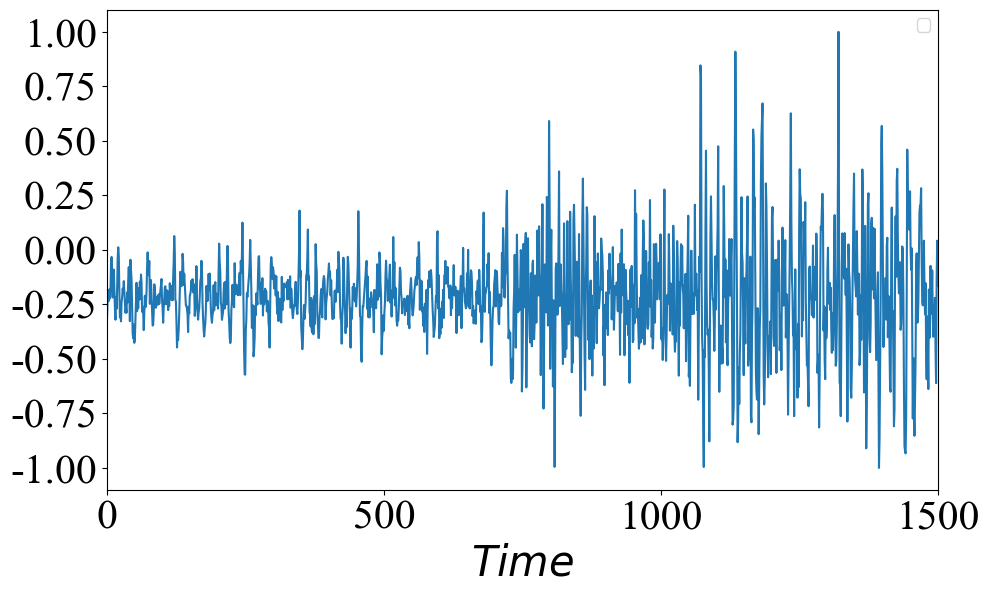} &
\includegraphics[width=0.23\linewidth]{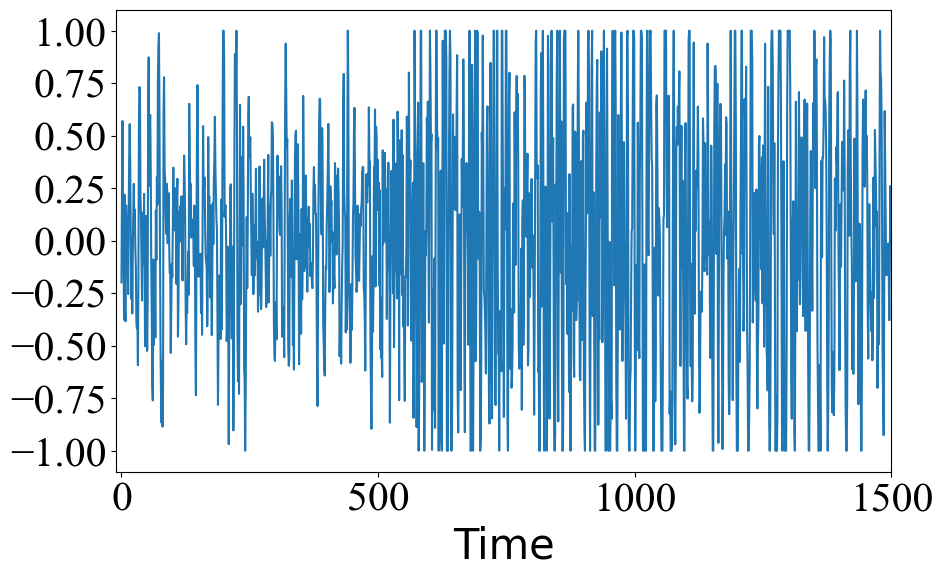} &
\includegraphics[width=0.23\linewidth]{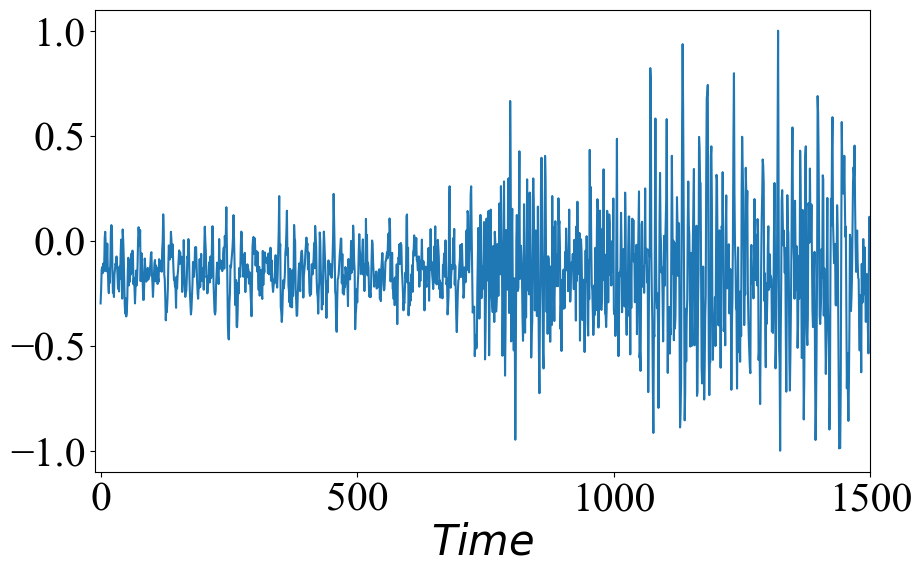} \\
\addlinespace[3pt]
(b)OM Indicator & 
\includegraphics[width=0.23\linewidth]{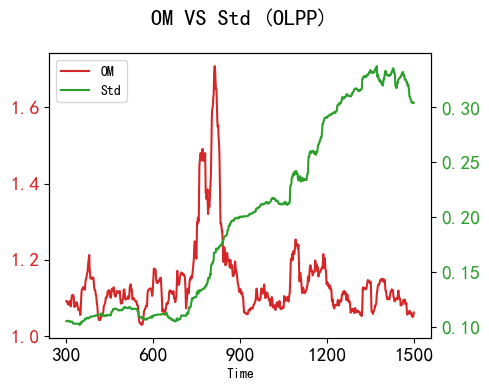} &
\includegraphics[width=0.23\linewidth]{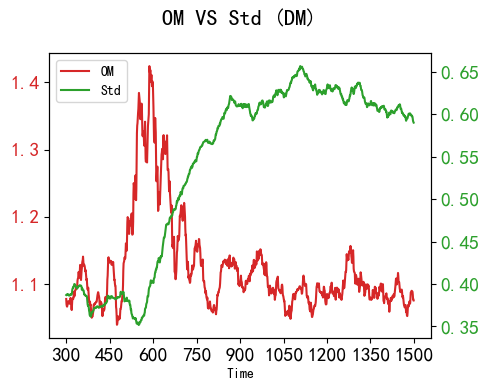} &
\includegraphics[width=0.23\linewidth]{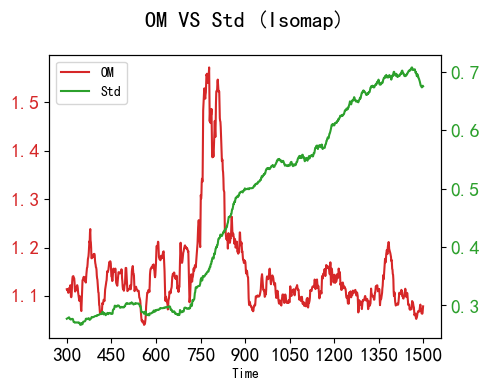} \\
\addlinespace[3pt]
(c)SF Indicator & 
\includegraphics[width=0.23\linewidth]{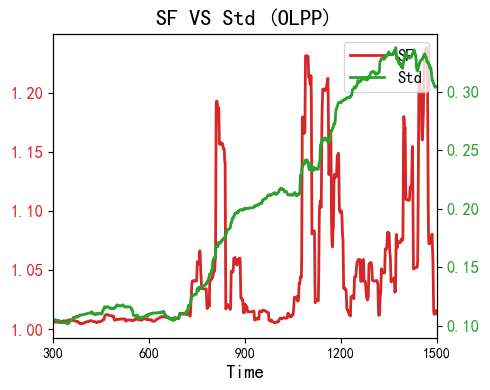} &
\includegraphics[width=0.23\linewidth]{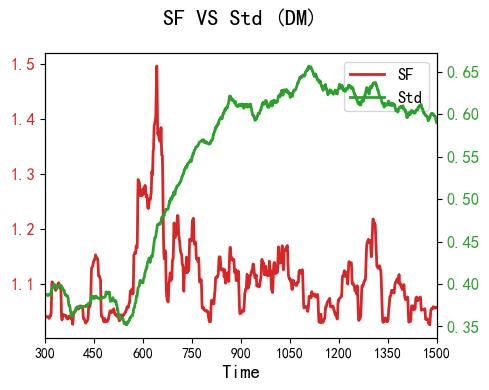} &
\includegraphics[width=0.23\linewidth]{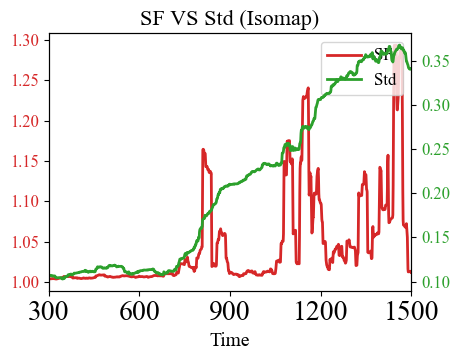} \\
\midrule
\textbf{Result/Method} & \textbf{NPE} & \textbf{PCA} & \textbf{Kernel PCA} \\
\midrule
(d)$\phi_2$ Data& 
\includegraphics[width=0.23\linewidth]{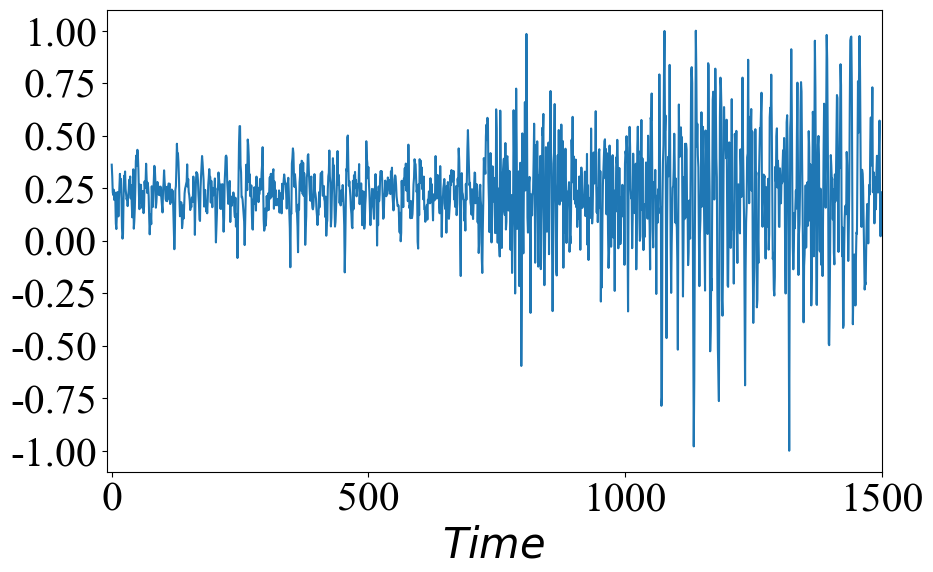} &
\includegraphics[width=0.23\linewidth]{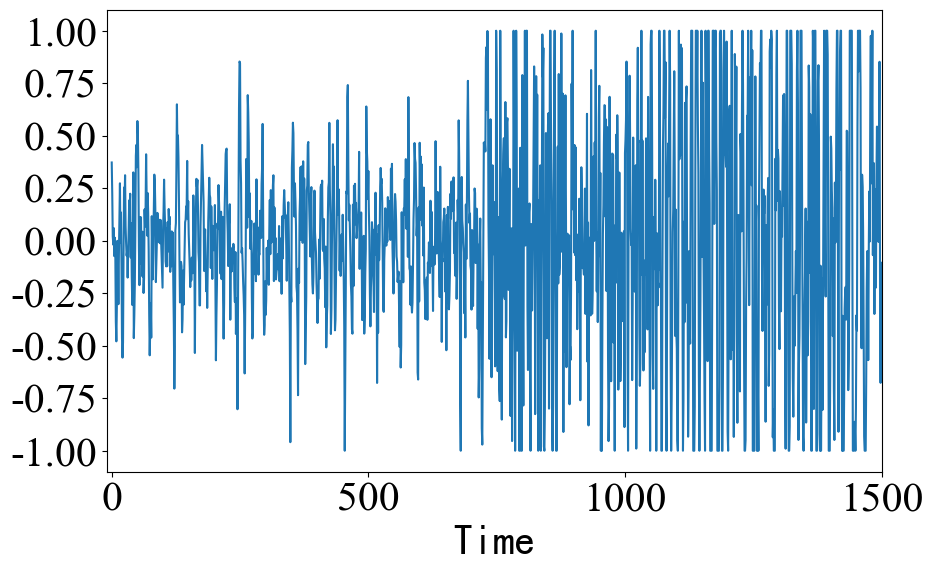} &
\includegraphics[width=0.23\linewidth]{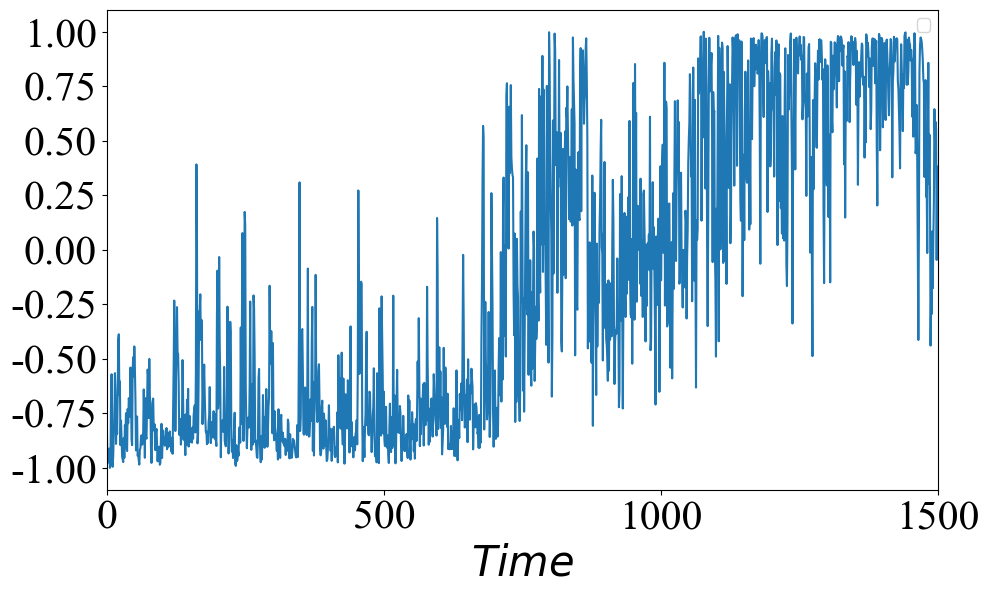} \\
\addlinespace[3pt]
(e)OM Indicator & 
\includegraphics[width=0.23\linewidth]{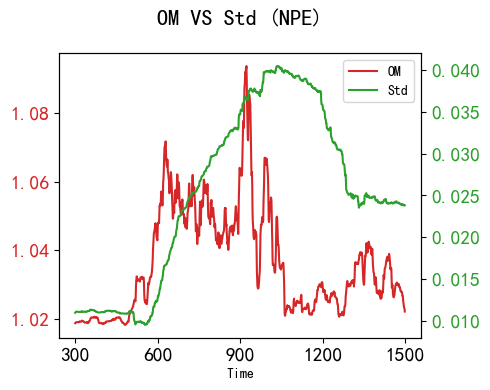} &
\includegraphics[width=0.23\linewidth]{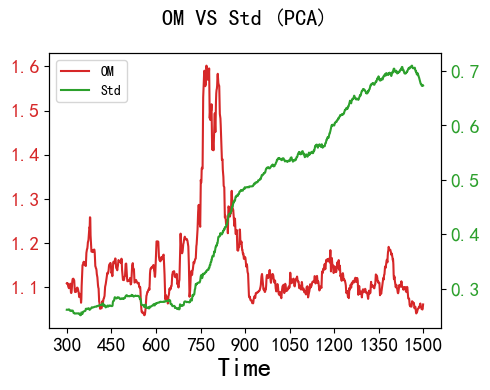} &
\includegraphics[width=0.23\linewidth]{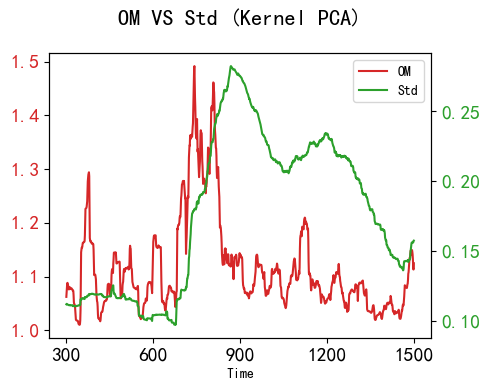} \\
\addlinespace[3pt]
(f)SF Indicator & 
\includegraphics[width=0.23\linewidth]{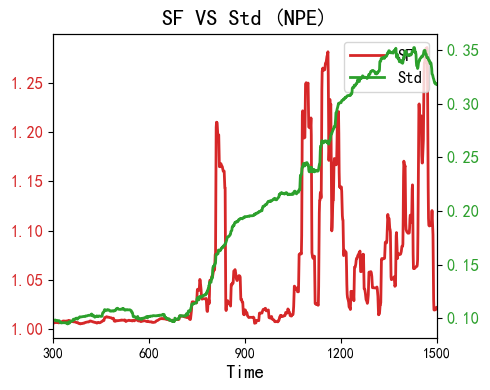} &
\includegraphics[width=0.23\linewidth]{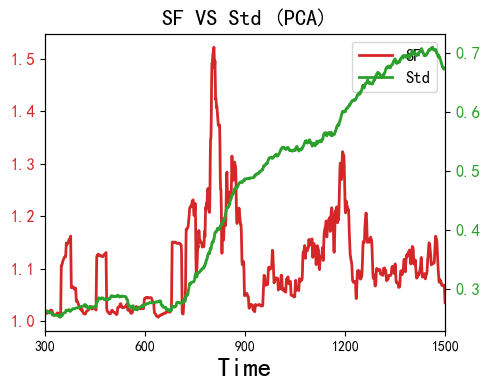} &
\includegraphics[width=0.23\linewidth]{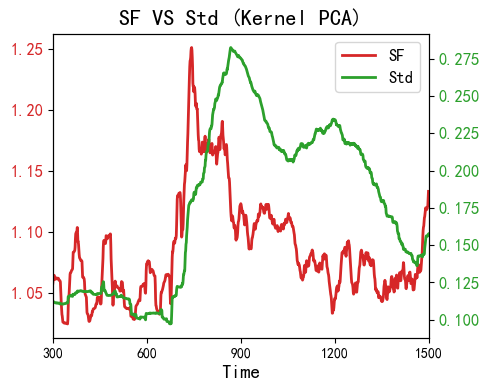} \\
\bottomrule
\multicolumn{4}{p{0.95\textwidth}}{%
  \textbf{} (a,d) The $\phi_2$ time series extracted via dimensionality reduction. (b,e) Comparison between the OM indicator (red) and the Std. (green). (c,f) Comparison between the SF indicator (red) and the Std (green). All indicators are computed using a sliding window (width l=300). Both the OM and SF indicators are derived within the framework of the previously learned SDE. The results show that both indicators issue an early-warning signal around \(T \approx 780\), demonstrating that the learned SDE model remains applicable to new data and can provide an effective stability warning.
} \\
\end{tabular}
\end{table}

Both warning indicators exhibited significant mutations or persistent supra-threshold phenomena around the 780th time point. This time point closely aligns with the clinically annotated onset of the epileptic seizure, and the warning signals appeared significantly earlier than the typical clinical symptoms of the ictal phase. These results strongly demonstrate that the early warning indicators constructed in this study, based on anisotropic dimensionality reduction and Schrödinger bridge theory, can effectively and preemptively detect the critical dynamic features of the brain's transition from the pre-ictal to the ictal state, validating the effectiveness and warning potential of the proposed method.

\section{Formal Error Analysis of the Early-warning Potential}
\hspace*{1em}To justify the reliability of the proposed SF indicator, we provide a high-probability error bound quantifying the deviation between the estimated $\widehat{\mathrm{SF}}(t)$ and its population counterpart $\mathrm{SF}^\ast(t)$. Our analysis builds upon the data processing pipeline in Figure~\ref{fig: Overall Framework Flowchart}: dimensionality reduction via Diffusion Map, SDE learning, and score function estimation through score matching.

\begin{assumption}
Let $\{X_t\}_{t\in[0,T]}$ be a stochastic process in $\mathbb{R}^D$ whose latent evolution lies near a smooth 1D manifold $\mathcal{M}$.
\begin{enumerate}[label=(A\arabic*), leftmargin=*, nosep]    
\item \textbf{(Manifold regularity)} $\mathcal{M}$ is $C^\infty$ and admits a $C^\infty$ embedding $\Phi_{\mathrm{true}}:\mathcal{M}\to\mathbb{R}$; the Diffusion Map estimator $\Phi_{\mathrm{DM}}$ satisfies            
\[             
\mathbb{E}\big[|\Phi_{\mathrm{DM}}(X_t) - \Phi_{\mathrm{true}}(X_t)|\big] \le \varepsilon_{\mathrm{man}} .            
\]            
\item \textbf{(SDE identifiability)} The true reduced dynamics $z_t = \Phi_{\mathrm{true}}(X_t)$ obeys            
\[            
dz_t = b^\ast(z_t)\,dt + \sigma^\ast(z_t)\,dW_t,            
\]
    with $b^\ast,\sigma^\ast\in C^\infty_b(\mathbb{R})$, and $0<\sigma_{\min}\le |\sigma^\ast(z)|\le\sigma_{\max}$ for all $z$.
    The estimated drift $\hat{b}$ and diffusion $\hat{\sigma}$ satisfy            
\[            
\| \hat{b}-b^\ast \|_\infty \le \varepsilon_b, \qquad \| \hat{\sigma}^2 - (\sigma^\ast)^2\|_\infty \le \varepsilon_\sigma .            
\]            
\item \textbf{(Density positivity)} The marginal density $p^\ast(z,t)>0$ and $\inf_{z,t} p^\ast(z,t)\ge p_{\min}>0$.            
\item \textbf{(Score regularity)} The true score $s^\ast(z,t)=\nabla_z\log p^\ast(z,t)$ belongs to Sobolev space $H^\beta([0,T]\times\mathbb{R})$ for some $\beta>1/2$, and satisfies $\|s^\ast\|_\infty\le S_{\max}$, $\|\nabla_z s^\ast\|_\infty\le L_s$.    
\item \textbf{(Score network)} $\hat{s}_\theta(z,t)$ is trained via score matching on $n = MN$ samples from $M$ SDE trajectories of length $N$, using a 3-layer MLP with hidden size $H$.
\end{enumerate}
\end{assumption}                                                                                                                                          

The following theorem quantifies how each stage of our pipeline contributes to SF uncertainty.

\begin{theorem}[High-probability error bound for Early-warning Potential]
\label{thm:EP_error}
Under the above assumptions, for any confidence level $\delta\in(0,1)$, with probability at least $1-\delta$ the estimated Early-warning Potential
\[
\widehat{\mathrm{SF}}(t) = \frac{1}{l}\int_{t-l}^{t}\frac{1}{n_\tau}\sum_{i=1}^{n_\tau}\|\hat{s}_\theta(z_i^{(\tau)},\tau)\|^2\,d\tau
\]
satisfies the uniform bound
\begin{align}
\sup_{t\in[l,T]} \big|\widehat{\mathrm{SF}}(t) - \mathrm{SF}^\ast(t)\big|
\le \underbrace{C_1\big(\varepsilon_{\mathrm{man}} + \varepsilon_b + \varepsilon_\sigma\big)}_{\text{manifold \& SDE error}}
+ \underbrace{C_2\big(\mathcal{E}_{\mathrm{app}} + \mathcal{E}_{\mathrm{est}} + \mathcal{E}_{\mathrm{opt}}\big)}_{\text{score approximation \& estimation error}}
+ \underbrace{ C_3/\sqrt{n} + C_4(\Delta t)^2}_{\text{sampling \& discretization}},
\label{eq:EP_error_bound}
\end{align}
where
\begin{align*}
\mathcal{E}_{\mathrm{app}} &\lesssim H^{-\beta}, &
\mathcal{E}_{\mathrm{est}} &\lesssim \sqrt{\frac{HL\log H + \log(2/\delta)}{n}}, &
\mathcal{E}_{\mathrm{opt}} &\le \varepsilon_{\mathrm{opt}}.
\end{align*}
The constants depend only on $S_{\max},L_s,p_{\min},\sigma_{\min},\sigma_{\max},T,l$, and are given explicitly by
\[
\begin{aligned}
C_1 &= \frac{2S_{\max}}{p_{\min}}\Bigl(1 + \frac{L_s T}{\sigma_{\min}^2}\Bigr),\\
C_2 &= 2S_{\max},\\
C_3 &= 2S_{\max}^2,\\
C_4 &= \frac{2}{3} \bigl(\|\partial_t s^\ast\|_\infty^2 + 2S_{\max} \|\partial_t\nabla s^\ast\|_\infty\bigr) l.
\end{aligned}
\]
\end{theorem}                                                                                                                                              

\begin{proof}
Define auxiliary quantities:
\begin{itemize}[leftmargin=*, nosep]    
\item $\widetilde{\mathrm{SF}}(t) := \frac{1}{l}\int_{t-l}^{t}\mathbb{E}_{\hat{z}_\tau\sim \hat{p}(\cdot,\tau)}\big[\|\hat{s}_\theta(\hat{z}_\tau,\tau)\|^2\big] d\tau$ — EP under estimated SDE $\hat{p}$ using estimated score;        
\item $\overline{\mathrm{SF}}(t) := \frac{1}{l}\int_{t-l}^{t}\mathbb{E}_{\hat{z}_\tau\sim \hat{p}(\cdot,\tau)}\big[\|s^\ast(\hat{z}_\tau,\tau)\|^2\big] d\tau$ — SF under estimated SDE but \emph{true} score;
\item $\mathrm{SF}^\ast(t) := \frac{1}{l}\int_{t-l}^{t}\mathbb{E}_{z_\tau\sim p^\ast(\cdot,\tau)}\big[\|s^\ast(z_\tau,\tau)\|^2\big] d\tau$ — population SF.
\end{itemize}

We decompose the total error as
\[
|\widehat{\mathrm{SF}}(t) - \mathrm{SF}^\ast(t)|
\le |\widehat{\mathrm{SF}} - \widetilde{\mathrm{SF}}|
+ |\widetilde{\mathrm{SF}} - \overline{\mathrm{SF}}|
+ |\overline{\mathrm{SF}} - \mathrm{SF}^\ast|.
\]

\noindent\textbf{Step (I): Sampling and Discretization Error.}
Let $z_i^{(\tau)}\sim\hat{p}(\cdot,\tau)$ i.i.d., and approximate the time integral by trapezoidal rule with step $\Delta t$.
Since $\|\hat{s}_\theta\| \le S_{\max} + \varepsilon_s \le 2S_{\max}$ with high probability, Hoeffding’s inequality yields for each $\tau$,
\[
\mathbb{P}\Bigl(\big|\tfrac{1}{n_\tau}\sum_i\|\hat{s}_\theta\|^2 - \mathbb{E}\|\hat{s}_\theta\|^2\big|\ge \eta\Bigr)
\le 2\exp\!\Bigl(-\tfrac{n_\tau \eta^2}{2(2S_{\max})^4}\Bigr).
\]
Applying a union bound over at most $T/\Delta t$ discretization points and integrating over $[t-l,t]$, we obtain with probability at least $1-\delta/3$,
\[
\sup_{t} |\widehat{\mathrm{SF}}(t)-\widetilde{\mathrm{SF}}(t)|
\le C_3/\sqrt{n} + C_4(\Delta t)^2,
\]
where $C_3 = 2S_{\max}^2$ and $C_4 = \frac{2}{3} \bigl(\|\partial_t s^\ast\|_\infty^2 + 2S_{\max} \|\partial_t\nabla s^\ast\|_\infty\bigr) l$.

\noindent\textbf{Step (II): Score Estimation Error.}
Using $|a^2 - b^2| \le |a-b|(|a|+|b|) \le 2S_{\max}\|\hat{s}_\theta - s^\ast\|$ (since both scores are bounded by $S_{\max}$ w.h.p.), we have
\[
\big|\widetilde{\mathrm{SF}}(t)-\overline{\mathrm{SF}}(t)\big|
\le 2S_{\max} \|\hat{s}_\theta - s^\ast\|_{L^2(\hat{p}\otimes\mathrm{Unif}[0,T])}.
\]
Standard score matching theory \cite{hu2024score, song2020sliced} implies that, with probability at least $1-\delta/3$,
\[
\|\hat{s}_\theta - s^\ast\|_{L^2} \le \mathcal{E}_{\mathrm{app}} + \mathcal{E}_{\mathrm{est}} + \mathcal{E}_{\mathrm{opt}},
\]
with $\mathcal{E}_{\mathrm{app}} \lesssim H^{-\beta}$ \cite{YAROTSKY2017103},
$\mathcal{E}_{\mathrm{est}} \lesssim \sqrt{(HL\log H + \log(3/\delta))/n}$ \cite{pmlr-v65-harvey17a},
and $\mathcal{E}_{\mathrm{opt}}$ the optimization residual.
Thus, $|\widetilde{\mathrm{SF}} - \overline{\mathrm{SF}}| \le C_2(\mathcal{E}_{\mathrm{app}} + \mathcal{E}_{\mathrm{est}} + \mathcal{E}_{\mathrm{opt}})$ with $C_2 = 2S_{\max}$.

\noindent\textbf{Step (III): Model Misspecification Error.}
Let $z_t = \Phi_{\mathrm{true}}(X_t)$ and $\hat{z}_t = \Phi_{\mathrm{DM}}(X_t)$.
\begin{enumerate}[leftmargin=*, nosep]        
\item \emph{Manifold error}: By (A1), $\mathbb{E}[|z_t - \hat{z}_t|] \le \varepsilon_{\mathrm{man}}$. Since $s^\ast(\cdot,\tau)$ is $L_s$-Lipschitz and $\|s^\ast\|_\infty \le S_{\max}$,
\[
\big| \|s^\ast(z_t,\tau)\|^2 - \|s^\ast(\hat{z}_t,\tau)\|^2 \big|     
\le 2S_{\max}L_s |z_t - \hat{z}_t|.
\]
Taking expectation under the natural coupling induced by $X_t$, and integrating over $\tau$, we obtain
\[
\big| \mathbb{E}_{p^\ast_\tau}[\|s^\ast\|^2] - \mathbb{E}_{\widetilde{p}_\tau}[\|s^\ast\|^2] \big|    
\le 2S_{\max}L_s \varepsilon_{\mathrm{man}}.
\]

\item \emph{SDE error}: Let $p^\ast(\cdot,\tau)$ and $\hat{p}(\cdot,\tau)$ denote the marginal laws of the true and estimated SDEs. Under (A2)–(A3), standard perturbation analysis of the Fokker–Planck equation (see e.g., \cite{pavliotis2014stochastic}) yields
\[
\|p^\ast(\cdot,\tau) - \hat{p}(\cdot,\tau)\|_{L^1} \le C_{\mathrm{FP}} (\varepsilon_b + \varepsilon_\sigma),
\quad C_{\mathrm{FP}} = \frac{T}{\sigma_{\min}^2}.
\]
Since $\|s^\ast\|^2 \le S_{\max}^2$, this implies
\[
\big| \mathbb{E}_{\hat{p}_\tau}[\|s^\ast\|^2] - \mathbb{E}_{p^\ast_\tau}[\|s^\ast\|^2] \big|    
\le S_{\max}^2 C_{\mathrm{FP}} (\varepsilon_b + \varepsilon_\sigma).
\]
However, the Early-warning Potential is defined as a Fisher information, which involves integration against the density. To control sensitivity to distributional perturbations, we invoke the positivity assumption (A3): $p^\ast(z,t) \ge p_{\min} > 0$. This allows us to bound the relative error in the expectation by a factor proportional to $1/p_{\min}$. Combining both effects and integrating over $\tau \in [t-l,t]$, we obtain
\[
\big|\overline{\mathrm{SF}}(t)-\mathrm{SF}^\ast(t)\big|
\le \frac{2S_{\max}}{p_{\min}}\Bigl(1 + \frac{L_s T}{\sigma_{\min}^2}\Bigr)(\varepsilon_{\mathrm{man}} + \varepsilon_b + \varepsilon_\sigma)
= C_1(\varepsilon_{\mathrm{man}} + \varepsilon_b + \varepsilon_\sigma).
\]
\end{enumerate}

Applying a union bound over the three error events (each controlled with failure probability at most $\delta/3$) and summing the bounds from Steps (I)–(III) yields the desired high-probability error bound~\eqref{eq:EP_error_bound}.
\end{proof}

\begin{remark}
The bound is tight in the sense that:
\begin{enumerate}
    \item $C_1$ blows up as $p_{\min}\to 0$, reflecting the well-known difficulty of score estimation near low-density regions (e.g., saddle points);
    \item $C_2$ scales linearly with $S_{\max}$, consistent with empirical observations that score magnitude peaks near critical transitions.
\end{enumerate}
Thus, Theorem~\ref{thm:EP_error} not only supports the empirical findings but also provides guidance for practical tuning (e.g., increasing $H$ or $n$ when SF variance is high).
\end{remark}

\section{Conclusion}

\hspace*{1em}This study proposes a critical early-warning framework that integrates manifold learning, stochastic system identification, and Schrödinger bridge theory, aiming to address the challenge of predicting state transitions in high-dimensional complex systems. The framework first employs manifold learning to reduce the dimensionality of high-dimensional observational signals, extracting directional characteristics of system evolution. Subsequently, a stochastic differential equation model is established using data-driven methods, and the score function characterizing the evolution of the system's probability density is robustly estimated with the aid of enhanced sampling techniques. On this basis, a novel SF indicator is proposed by formulating a Schrödinger bridge problem. This indicator quantifies the likelihood of significant state transitions in a probabilistic form and employs a proportional calculation method consistent with the OM functional, significantly enhancing the detection sensitivity to subtle early-stage perturbations in the system.

Experimental results demonstrate that, in the task of epilepsy seizure warning, the SF indicator exhibits superior overall performance compared to traditional methods. It not only identifies the primary critical point earlier and more clearly but also effectively captures specific dynamic fluctuations across different stages before and after seizures, even under low-noise conditions. Notably, this framework possesses broad applicability. Beyond neurodynamic systems, the method is equally applicable to other complex systems exhibiting critical transition phenomena, such as tropical rainforest desertification, the emergence of financial "black swan events," and the onset of strokes. In these systems, high-dimensionality, nonlinearity, and non-equilibrium characteristics are prevalent, and the proposed data-driven integrative framework offers a novel approach to extracting robust early-warning signals from complex observational data.

Looking ahead, the proposed framework can be extended to model a wider range of stochastic processes. An important direction is generalizing the current stochastic differential equation model based on Gaussian noise to systems driven by Lévy noise with heavy-tailed characteristics. Such models can more accurately describe severe fluctuations and discontinuous jump behaviors observed in systems like stock markets or extreme climate events. By developing corresponding methods for score function estimation and Schrödinger bridge solving, it is expected to achieve more precise identification of critical states in fields such as financial risk warning and extreme event prediction, thereby enhancing the framework's early-warning capability for strongly non-Gaussian, highly volatile complex systems.

In summary, this work constructs a complete analytical workflow for identifying critical transitions from high-dimensional data. The proposed SF indicator and modeling framework have not only proven effective in specific neuroscience applications but also provide a generalizable theoretical tool and methodological foundation for critical early warning in a broader range of complex system types. Furthermore, it opens new avenues for extending early-warning research towards non-Gaussian dynamics.

\section*{Acknowledgements}

\hspace*{1em}This work was partly supported by the NSFC  International Collaboration Fund for Creative Research Teams (Grant W2541005), the Guangdong Provincial Key Laboratory of Mathematical and Neural Dynamical Systems (Grant 2024B1212010004), the Cross Disciplinary Research Team on Data Science and Intelligent Medicine (2023KCXTD054), the Guangdong-Dongguan Joint Research Fund (Grant 2023A151514 0016), Dongguan Key Laboratory for Data Science and Intelligent Medicine, and the Guangdong-Dongguan Joint Research Grant (2023A1515140016).

%\subsection*{Funding}
%Name financially supporting bodies (written out in full), followed by the funding awardee and associated grant numbers (if applicable) in square brackets. 

\printbibliography

\end{document}